\documentclass[11pt]{article}

\usepackage[letterpaper,top=2cm,bottom=2cm,left=3cm,right=3cm,marginparwidth=1.75cm]{geometry}

\usepackage{extpfeil}
\usepackage{microtype}
\usepackage{graphicx}
\usepackage{subfigure}
\usepackage{booktabs} 
\usepackage{enumitem}
\usepackage[colorlinks=true, allcolors=blue]{hyperref}
\usepackage{bm}
\usepackage{ulem}
\usepackage{physics}
\usepackage{algorithm}
\usepackage{algorithmic}

\usepackage{bbm}
\usepackage{amsmath}
\usepackage{amssymb}
\usepackage{mathtools}
\usepackage{amsthm}

\usepackage{braket}

\usepackage{fullpage}

\theoremstyle{plain}
\newtheorem{theorem}{Theorem}[section]

\newtheorem{lemma}[theorem]{Lemma}
\newtheorem{corollary}[theorem]{Corollary}
\newtheorem{statement}[theorem]{Statement}
\theoremstyle{definition}
\newtheorem{definition}[theorem]{Definition}

\newtheorem{task}[theorem]{Task}
\theoremstyle{remark}
\newtheorem{remark}[theorem]{Remark}

\newcommand{\eqn}[1]{(\ref{eqn:#1})}
\newcommand{\eq}[1]{(\ref{eq:#1})}

\newcommand{\thm}[1]{\hyperref[thm:#1]{Theorem~\ref*{thm:#1}}}
\newcommand{\cor}[1]{\hyperref[cor:#1]{Corollary~\ref*{cor:#1}}}
\newcommand{\define}[1]{\hyperref[def:#1]{Definition~\ref*{def:#1}}}
\newcommand{\lem}[1]{\hyperref[lem:#1]{Lemma~\ref*{lem:#1}}}
\newcommand{\prop}[1]{\hyperref[prop:#1]{Proposition~\ref*{prop:#1}}}
\newcommand{\prob}[1]{\hyperref[prob:#1]{Problem~\ref*{prob:#1}}}
\newcommand{\assum}[1]{\hyperref[assum:#1]{Assumption~\ref*{assum:#1}}}
\newcommand{\fig}[1]{\hyperref[fig:#1]{Figure~\ref*{fig:#1}}}
\newcommand{\tab}[1]{\hyperref[tab:#1]{Table~\ref*{tab:#1}}}
\newcommand{\alg}[1]{\hyperref[alg:#1]{Algorithm~\ref*{alg:#1}}}
\renewcommand{\sec}[1]{\hyperref[sec:#1]{Section~\ref*{sec:#1}}}
\newcommand{\de}[1]{\hyperref[def:#1]{Definition~\ref*{def:#1}}}
\newcommand{\append}[1]{\hyperref[append:#1]{Appendix~\ref*{append:#1}}}
\newcommand{\fac}[1]{\hyperref[fac:#1]{Fact~\ref*{fac:#1}}}
\newcommand{\lin}[1]{\hyperref[lin:#1]{Line~\ref*{lin:#1}}}

\newcommand{\sta}[1]{\hyperref[sta:#1]{Statement~\ref*{sta:#1}}}
\newcommand{\tsk}[1]{\hyperref[task:#1]{Task~\ref*{task:#1}}}
\newcommand{\rem}[1]{\hyperref[rem:#1]{Remark~\ref*{rem:#1}}}

\DeclareMathOperator{\CNOT}{CNOT}
\DeclareMathOperator{\MYNOT}{NOT}
\DeclareMathOperator{\CPINOT}{C_{\Pi}NOT}
\DeclareMathOperator{\Amp}{AmpEst}
\DeclareMathOperator{\Fix}{FixSearch}

\newcommand{\E}{\mathbb{E}}

\renewcommand{\d}{\mathrm{d}}

\renewcommand\bra[1]{{\langle{#1}|}}
\makeatletter
\renewcommand\ket[1]{%
  \@ifnextchar\bra{\k@t{#1}\!}{\k@t{#1}}%
}
\newcommand\k@t[1]{{|{#1}\rangle}}
\makeatother

\title{ Quantum Non-Identical Mean Estimation:\\Efficient Algorithms and Fundamental Limits}
\author{\hspace{-4mm}Jiachen Hu\thanks{Peking University} \quad Tongyang Li$^*$ \quad Xinzhao Wang$^*$ \quad Yecheng Xue$^*$ \quad Chenyi Zhang\thanks{Stanford University} \quad Han Zhong$^*$}
\date{}

\begin{document}

\maketitle

\begin{abstract}
We systematically investigate quantum algorithms and lower bounds for mean estimation given query access to non-identically distributed samples. 
On the one hand, we give quantum mean estimators with quadratic quantum speed-up given samples from different bounded or sub-Gaussian random variables. On the other hand, we prove that, in general, it is impossible for any quantum algorithm to achieve quadratic speed-up over the number of classical samples needed to estimate the mean $\mu$, where the samples come from different random variables with mean close to $\mu$. Technically, our quantum algorithms reduce bounded and sub-Gaussian random variables to the Bernoulli case, and use an uncomputation trick to overcome the challenge that direct amplitude estimation does not work with non-identical query access. Our quantum query lower bounds are established by simulating non-identical oracles by parallel oracles, and also by an adversarial method with non-identical oracles. Both results pave the way for proving quantum query lower bounds with non-identical oracles in general, which may be of independent interest.
\end{abstract}


\section{Introduction}
\label{sec:intro}
The problem of estimating the mean $\mu$ of a random variable $X$ given its i.i.d.~samples is a fundamental problem in statistics. For any random variable $X$ with finite variance $\sigma^2$, the median-of-means estimator can estimate $\mu$ to within additive error $\epsilon$ with failure probability $\le\delta$ using $O(\frac{\sigma^2}{\epsilon^2}\log(\frac{1}{\delta}))$ samples. This sample complexity is known to be tight up to a constant multiplicative factor \cite{dagum2000optimal}. 

On the other hand, suppose that a quantum computer has access to a unitary $U$ and its inverse such that $U\ket{\mathbf{0}}$ encodes the random variable $X$ coherently, and each application of $U$ and $U^{\dagger}$ as a black-box oracle can be regarded as a quantum analogue of getting a sample of the random variable $X$. Therefore, the application of $U$ is sometimes called a \textit{quantum experiment} \cite{hamoudi2021quantum}. Under this assumption, a quantum computer can estimate the mean of $X$ with $O(\frac{\sigma}{\epsilon}\log(\frac{1}{\delta}))$ quantum experiments \cite{montanaro2015quantum}, which achieves quadratic speed-up compared to the classical counterpart. Such quantum mean estimators embrace various applications, including approximate counting \cite{montanaro2015quantum,cornelissen2023sublinear}, data stream estimation \cite{hamoudi2019quantum}, derivative pricing in finance \cite{chakrabarti2021threshold}, etc.

In some cases, we are interested in estimating the mean of ``close" random samples, such as random samples with the same mean but different distributions. For example, it is ubiquitous that the measurements of random samples have small systematic errors. In such cases there may be small difference between the means of the actual distributions of the measured random samples, and our algorithms and lower bounds also take this into account. One specific example is to learn a linear system discussed below.
In classical mean estimation, the same method for identical random variables also works for non-identical random variables.
As long as the variance of all random variables is bounded by $\sigma^2$, the median-of-means estimator can be directly adapted to these situations
, yielding an algorithm with the same complexity. However, it is unclear whether similar results hold in the regime of quantum mean estimation.
Therefore, it is a natural question whether we can achieve quantum speed-up for the mean estimation problem with non-identically distributed samples. 

Below we provide a potential application for the quantum mean estimation with non-identically distributed samples.

\paragraph*{Quantum Linear System}

A classical linear dynamical system (LDS) is defined as 
\begin{align}
x_{t+1} = Ax_t + w_t,\ x_t \in \mathbb{R}^n,\ w_t \sim \mathcal{N}(\bm{0}, \sigma_w^2),\ \|A\|_2 < 1, x_1 = \bm{0}
\end{align}
where $x_t$ is the state at time step $t$, and $w_t$ is a random noise at step $t$. A well-known problem in LDS is to do the system identification: estimating the transition matrix $A$ given a series of states starting from step $1$.
The standard approach to estimate transition matrix $A$ in the classical linear system is ordinary least squares (OLS) \cite{dean2020sample, simchowitz2020naive}. 

Consider the quantum counterpart of LDS (for example, when simulating a LDS on a quantum computer): 
\begin{align}
    U_f|\psi_x\rangle |0\rangle & = \int_{\mathbb{R}^n} \sqrt{f_w(w)}|\psi_x\rangle  |\psi_{Ax + w}\rangle \mathrm{d}w,\\
    U_o|\psi_x\rangle |0\rangle & = |\psi_x\rangle |x\rangle,
\end{align}
here $f_w(w)$ is the probability density function (pdf) of $\mathcal{N}(\bm{0}, \sigma_w^2)$, and $|\psi_x\rangle$ is an arbitrary embedding of the raw state $x$. It is natural to ask whether it is possible to estimate $A$ by a quantum algorithm with desired speed-up in quantum linear systems. 
Actually, it is indeed possible 
with a procedure presented in \sec{implication_quantum_ls}. This estimation procedure uses multiple calls to $U_f$ to construct a new oracle $U_{t_0}$ for some step $t_0$, which encodes a probability distribution over the matrix space with $A$ as the mean value. However, the distribution encoded by $U_{t_0}$ is different for different $t_0$, though their means are all equal to $A$. Therefore, this problem presents another motivation of the quantum non-identical mean estimation problem.

In general, the quantum linear system problem described above is a special class of quantum estimation problem in which quantum probability oracles have a time-varying zero-mean noise. The distribution of noise at each step is different but all zero-mean. The number of samples at each step is limited. 

\subsection{Contributions}
In this paper, we systematically analyze the sample complexity of the \textit{quantum non-identical mean estimation problem} (see its formal definition in \tsk{mean-estimation}). Roughly speaking, the quantum algorithm is given $T$ different random variables in turn and can get $m\in \mathbb{N}$ samples from each random variable. Suppose that the mean of every random variable is in $(\mu-c\epsilon,\mu+c\epsilon)$ for some constant $0<c<1$, the quantum non-identical mean estimation problem is to estimate $\mu$ up to additive error $\epsilon$. If all random variables are bounded or sub-Gaussian (see definition in \de{sub-gaussian}), for accuracy $\epsilon$ and $m=\Omega(\log(\frac{1}{\epsilon}))$, we give quantum algorithms solving the quantum non-identical mean estimation problem with quadratic speed-up.
\begin{theorem}[Informal versions of \thm{mean-bounded} and \thm{mean-sub-gaussian}]
    \label{thm:upper-bound-rough}
  For the quantum non-identical mean estimation problem with sufficiently small accuracy $\epsilon$, \begin{itemize}
    \item if all random variables are bounded in $[L,H]$ and  $m=\Omega(\log(\frac{H-L}{\epsilon}))$, there is a quantum algorithm that estimates $\mu$ to within additive error $\epsilon$ if $T = \Omega(\frac{H-L}{\epsilon})$. The algorithm uses $O(\frac{H-L}{\epsilon}\log(\frac{H-L}{\epsilon}))$ samples in total; 
    \item if all random variables are sub-Gaussian with parameter $K$ and $m=\tilde{\Omega}(\log(\frac{K}{\epsilon}))$, there is a quantum algorithm that estimates $\mu$ to within additive error $\epsilon$ if $T = \tilde{\Omega}(\frac{K}{\epsilon})$. The algorithm uses $\tilde{O}(\frac{K}{\epsilon})$ samples in total.
  \end{itemize}
\end{theorem}
In the worst case, the variance of random variables bounded in $[L,H]$ can be $(H-L)^2/4$, so the optimal classical estimator needs $\Theta((H-L)^2/\epsilon^2)$ samples to estimate $\mu$ up to additive error $\epsilon$. For normal random variables, their sub-Gaussian parameter $K$ equals their standard deviation $\sigma$, so the optimal classical estimator needs $\Theta(K^2/\epsilon^2)$ samples to estimate $\mu$ up to additive error $\epsilon$. Therefore, the quantum estimators in \thm{upper-bound-rough} achieve nearly quadratic speed-up compared to classical estimators.

On the other hand, for $m=1$, we show that any algorithm with relatively small working register have no speed-up compared to classical estimators. 
\begin{theorem}[Informal version of \thm{mean-estimation-lower-0}]\label{thm:mean-estimation-lower-0-informal}
  Suppose all random variables in the quantum non-identical mean estimation problem with $m=1$ have mean bounded by $R$ and variance bounded by $\sigma^2$. Let $\mathcal{A}$ be a quantum query algorithm acting on query register $Q$, working register $W$ such that the number of qubits in $Q$ is larger than that in $W$ by $\Omega(\log(\frac{R}{\epsilon}))$. It requires $T = \Omega(\frac{\sigma^2}{\epsilon^2})$ if there exists an algorithm $\mathcal{A}$ solving this problem. The sample complexity of $\mathcal{A}$ is $T = \Omega(\frac{\sigma^2}{\epsilon^2})$.
\end{theorem}

For general $m\ge 1$, we give another sample complexity lower bound of estimating mean of Bernoulli random variables.
\begin{theorem}[Informal version of \thm{mean-estimation-lower-1}]\label{thm:mean-estimation-lower-1-informal}
    \label{thm:mean-estimation-lower-1-rough}
   Suppose all random variables in the quantum non-identical mean estimation problem with $m\ge 1$ are Bernoulli random variables with mean $\mu \in(0,1)$, and the accuracy $\epsilon$ satisfies $\epsilon \le \mu(1-\mu)$ and $\epsilon=O(\frac{1}{m^2})$. It requires $T = \Omega(\frac{1}{\epsilon m^2})$ if there exists a quantum query algorithm solving this problem. The sample complexity is $mT = \Omega(\frac{1}{\epsilon m})$ in total.
\end{theorem}
In \thm{mean-estimation-lower-1-rough}, we take the Bernoulli random variables as a hard instance for the quantum non-identical mean estimation problem. 
Note that if $\epsilon=\Theta(\mu(1-\mu))$, the classical optimal estimator needs $\Theta(\frac{\mu(1-\mu)}{\epsilon^2})=\Theta(\frac{1}{\epsilon})$ samples to estimate the mean of the Bernoulli random variable. Therefore, \thm{mean-estimation-lower-1-rough} shows that there is no quantum speed-up  in this case if $m=O(1)$. However, it does not rule out the possibility of quantum speed-up for estimating the mean of Bernoulli random variables with $\epsilon=o(\mu(1-\mu))$ or $m=\Omega(1)$. For example, if $\mu=\Theta(1),\epsilon=o(1)$, and $m=\Omega(\log(\frac{1}{\epsilon}))$, the quantum estimator for bounded random variables in \thm{upper-bound-rough} can estimate $\mu$ up to error $\epsilon$ using $O(\frac{1}{\epsilon}\log(\frac{1}{\epsilon}))$ samples while classical estimators need $\Omega(\frac{1}{\epsilon^2})$ samples.

In addition, \thm{mean-estimation-lower-0-informal} and \thm{mean-estimation-lower-1-informal} give two different lower bounds when $m=1$. Compared with \thm{mean-estimation-lower-1-informal}, the lower bound in \thm{mean-estimation-lower-0-informal} matches the classical upper bound for general distributions with variance $\sigma^2$, but an additional requirement is that the register $W$ has relatively small dimension.

Finally, we use Bernoulli random variable as an example to summary our systematical investigation on the quantum non-identical mean estimation problem.
\begin{corollary}
For Bernoulli random variable with mean $\mu$ such that $\epsilon = \Theta(\mu(1-\mu))$, \begin{itemize}
    \item if $m = \Omega(\log(1/\epsilon) )$ and $T = \Omega(1/\epsilon)$, there exists an algorithm solving this problem using $O(\frac{1}{\epsilon}\log(1/\epsilon))$ quantum samples, achieving a near-quadratic speed-up;
    \item if $m = \Omega(\log(1/\epsilon))$ and $T = o(1/\epsilon m^2)$, there is no quantum algorithm solving this problem. There is an additional requirement that $\epsilon = O(1/m^2)$;
    \item if $m = O(1)$, there is no quantum speed-up for this problem.
\end{itemize}
\end{corollary}
\begin{proof}
    This corollary comes directly from \thm{upper-bound-rough},  \thm{mean-estimation-lower-0-informal}, and  \thm{mean-estimation-lower-1-informal}.
\end{proof}

\subsection{Techniques}
\subsubsection{Upper Bound}
From a high-level perspective, our quantum algorithms for non-identical mean estimation
encode the mean to an amplitude, use an uncomputation trick to be introduced below to align different oracles, and then use amplitude estimation to estimate the mean. 

We start with the bounded case. Recall that this paper studies non-identically distributed samples and assumes that we have access to unitaries  $O_{X_1},\ldots,O_{X_T}$, where
\begin{align}
O_{X_i}\ket{\mathbf{0}} = \sum_{x \in E_i} \sqrt{p_i(x)}\ket{\psi_x^{(i)}}\ket{x}.
\end{align}
The mean $\mu = \mu_i = \sum_{x \in E_i}p_i(x)x$ is equal for different $i \in [T]$ (In fact, these $\mu_i$ can be slightly different -- see \rem{delta} for more details), but each $O_{X_i}$ has potentially different garbage states $\ket{\psi_x^{(i)}}$ and each can only be used for very limited times. Suppose that for any $i \in [T]$, the bounded random variable $X_i$ satisfies $X_i\in [L,H]$. If we have sufficient access to any specific $O_{X_i}$, 
we can construct a unitary
\begin{align}
U_i|\mathbf{0}\rangle|0\rangle = \sqrt{q}|\psi_1^{(i)}\rangle|1\rangle + \sqrt{1-q}|\psi_0^{(i)}\rangle|0\rangle
\end{align}
by one call to $O_{X_i}$ and a series of controlled rotations \cite{montanaro2015quantum}, where $q = (\mu -L)/(H-L)$. Consequently, the mean is encoded to an amplitude and direct amplitude estimation provides mean estimation with quadratic quantum speedup. However, in the non-identical case, we do not have sufficient number of calls to any specific $U_i$ to provide quadratic speedup. Furthermore, it is very difficult to use a mixture of different $U_i$ in amplitude estimation \cite{brassard2002quantum}. This is due to the reason that amplitude estimation is based on Grover's algorithm \cite{grover1996fast}, which is essentially rotation in a two-dimensional plane spanned by two specific quantum states related to $U_i$. In our case, different $U_i$ may have different $\ket{\phi_1^{(i)}}$ and $\ket{\phi_0^{(i)}}$, which forms different rotation planes and thus their mixed use is invalid. However, we can use a small number of calls to $U_i$ to construct a unitary such that
\begin{align}
S_i\ket{\mathbf{0}}=\sqrt{1-\epsilon_i}\ket{0}\Bigl(\sqrt{r}\ket{1}+
\sqrt{1-r}\ket{0}\Bigr) + \sqrt{\epsilon_i}\ket{1}\ket{\mathrm{garbage}_i}
\end{align}
with $r$ being a bijective function of $q$ (the concrete value to be shown later) and $\epsilon_i$ being sufficiently small. Since the garbage state is small enough to be handled as an approximation error, $S_i$ can be seen as an approximation of an unitary $S\colon \ket{0}\rightarrow \sqrt{r}\ket{1} + \sqrt{1-r}\ket{0}$. Therefore, We can then use these $S_i$ instead of $S$ to perform amplitude estimation, which provides estimation for $r$ and thus $q$ and $\mu$.

The construction of $S_i$ can be accomplished by an uncomputation trick \cite{cornelissen2023sublinear} and fixed-point search \cite{yoder2014fixed}. Specifically, the uncomputation trick is to perform a unitary 
\begin{align}
V_i = (U_i^{\dagger}\otimes I)(I\otimes \CNOT) (U_i\otimes I)
\end{align}
instead of $U_i$, which enjoys a property that it extracts the value of $q$ separated from a garbage state related to $\ket{\phi_1^{(i)}}$ and $\ket{\phi_0^{(i)}}$. The computing result of $\bra{b}\bra{0}\bra{\mathbf{0}} V_i \ket{\mathbf{0}}\ket{0}\ket{0}$ for $b \in \{0,1\}$ tells that $V_i \ket{\mathbf{0}}\ket{0}\ket{0}$ only has components $\ket{\mathbf{0}}\ket{0}\ket{0}$, $\ket{\mathbf{0}}\ket{0}\ket{1}$, and a garbage state orthogonal to them. Besides, the amplitudes of the first two components are determined by $q$. In particular, it satisfies
 \begin{align}
    V_i\ket{\mathbf{0}}\ket{0}\ket{0} 
    & = \sqrt{2q^2-2q+1}\ket{\mathbf{0}}\ket{0}\biggl(\frac{q}{\sqrt{2q^2-2q+1}}\ket{1}+\frac{1-q}{\sqrt{2q^2-2q+1}}\ket{0}\biggr)\nonumber\\
    &\qquad\qquad+ \sqrt{2q-2q^2}\ket{\mathrm{garbage}_i},
  \end{align}
  where $\ket{\mathrm{garbage}_i}$ is a unit garbage state and $(I\otimes\bra{0}\bra{\mathbf{0}})|\mathrm{garbage}_i\rangle = 0$. Therefore, we can use fixed-point quantum search \cite{yoder2014fixed} to stably amplify the amplitude of the state $\frac{q}{\sqrt{2q^2-2q+1}}\ket{1}+\frac{1-q}{\sqrt{2q^2-2q+1}}\ket{0}$ and thus $S_i$ is constructed with $r = \frac{q^2}{2q^2-2q+1}$. See \thm{mean-bounded} for more details.

For a sub-Gaussian random variable with the absolute value of mean bounded by the sub-Gaussian parameter $K$, the probability of the random variable being more than a threshold related to $K$ is sufficiently small and the mean of a truncated random variable can be a good enough approximation. Therefore, this case can be reduced to the case of bounded random variables. For general sub-Gaussian random variables $X_1,\ldots, X_T$, a constant number of classical experiments provide an estimation $\hat{\mu}$ within $K$-additive error, thus $X_1-\hat{\mu},\ldots, X_T-\hat{\mu}$ are sub-Gaussian random variables with the absolute value of mean bounded by $K$, which has been solved (see \thm{mean-sub-gaussian} for more details).

\subsubsection{Lower Bound}
We prove our two quantum query lower bounds using different techniques: the case $m=1$ (\thm{mean-estimation-lower-0-informal}) is established by simulating non-identical oracles by parallel oracles, and the case $m\geq 1$ (\thm{mean-estimation-lower-1-informal}) is established by an adversarial method with non-identical oracles. 

\paragraph*{Simulating $T$ Non-Identical Oracles by Constant $T$-Parallel Oracles} For the quantum non-identical mean estimation problem with $m=1$, we give a sample complexity lower bound in \thm{mean-estimation-lower-0} by constructing a quantum circuit with constant query depth simulating the original quantum circuit querying non-identical oracles. For any quantum query algorithm $\mathcal{A}$ using the \textit{state preparation oracle} $U_x$ such that the state $U_x\ket{\mathbf{0}}$ encodes the input, suppose that there is a sequence of unitary oracles that maps $\ket{\mathbf{0}}$ to the same state but have different effects acting on other states orthogonal to $\ket{\mathbf{0}}$. Suppose that the working register of $\mathcal{A}$ is relatively small and $\mathcal{A}$ queries $T$ non-identical oracles. In \thm{low-depth}, we prove that for any projection $\Pi$ with small image space, there is a quantum algorithm $\mathcal{A}'$ using two $T$-parallel queries such that \begin{align}
    \|\Pi\mathcal{A}\ket{\mathbf{0}}\|^2=\|(\Pi\otimes \langle\mathbf{0}|)\mathcal{A}'\ket{\mathbf{0}}\ket{\mathbf{0}}\|^2,
\end{align} where a $T$-parallel query is to query $T$ oracles simultaneously. This theorem builds a bridge between quantum algorithms with non-identical state preparation oracles and quantum algorithms with low query depth. If for any input $x$ correct outputs of $\mathcal{A}$ lie in a small space $V_x$, and let $\mathrm{Im}(\Pi)=V_x$, then \thm{low-depth} shows that $\mathcal{A}$ and $\mathcal{A'}$ have the same probability to output a correct answer. 

In \thm{mean-estimation-lower-0}, we prove that any quantum query algorithm $\mathcal{A}$ starting from an efficiently preparable state $\ket{\mathbf{0}}$ can be modified to recover the query register to $\ket{\mathbf{0}}$ with a small overhead. This reduces the dimension of the subspace that the correct outputs of $\mathcal{A}$ lie in, and then we use \thm{low-depth} to give a sample complexity lower bound of the quantum non-identical mean estimation problem with $m=1$ based on the facts that parallelization only brings classical advantage to solving the quantum approximate counting problem \cite{burchard2019lower}, and the quantum approximate counting problem can be reduced to estimating the mean of Bernoulli random variables. 

\paragraph*{Adversarial Method with Non-Identical Oracles} Given a boolean function $f\colon\{0,1\}^n \to \{0,1\}$ and access to a unitary oracle $O_x$ which encodes the information of some $x\in \{0,1\}^n$, the \textit{generalized adversarial method} \cite{hoyer2007negative} gives a tight query complexity lower bound of computing $f(x)$. For any quantum query algorithm $\mathcal{A}$ and $x\in \{0,1\}^n$, let $\ket{\psi_x^{(t)}}$ be the quantum state after $\mathcal{A}$ queries $O_x$ for $t$ times. Suppose $\mathcal{A}$ can compute $f(x)$ with high probability for all $x\in\{0,1\}^n$ using $T$ queries, then we have $\bra{\psi_x^{(T)}}\psi_y^{(T)}\rangle=1-\Omega(1)$ for all $x\in f^{-1}(0)$ and $y\in f^{-1}(1)$. Since $\bra{\psi_x^{(0)}}\psi_y^{(0)}\rangle=1$, to give a lower bound of $T$, it suffices to give an upper bound on the \textit{progress} at time $t$, $\bra{\psi_x^{(t-1)}}\psi_y^{(t-1)}\rangle-\bra{\psi_x^{(t)}}\psi_y^{(t)}\rangle$, for all $x\in f^{-1}(0)$, $y\in f^{-1}(1)$, and $t\in[T]$. The generalized adversarial method assigns a weight $\Gamma_{xy}$ to every pair of $x\in f^{-1}(0)$,\ $y\in f^{-1}(1)$, which proves an upper bound for the weighted progress at time $t$: \begin{align}
    S_{t-1}-S_{t}=\sum_{x\in f^{-1}(0),\ y\in f^{-1}(1)}\Gamma_{xy}(\bra{\psi_x^{(t-1)}}\psi_y^{(t-1)}\rangle-\bra{\psi_x^{(t)}}\psi_y^{(t)}\rangle),
\end{align}
and hence gives a lower bound on $T$. However, they regard $\ket{\psi_x^{(t-1)}}, \ket{\psi_y^{(t-1)}}$ as free variables independent of previous states $\ket{\psi_x^{(t')}}, \ket{\psi_y^{(t')}}$ for $t' < t-1$ while bounding the weighted progress at $t$, so their upper bound of $S_{t-1}-S_{t}$ is independent of $t$. Therefore, if the algorithm queries different oracles at different times, the adversarial method cannot give better lower bound than the case that all oracles are the same. In \lem{counting-non-identical}, we apply the adversarial method on the quantum approximate counting problem, but analyze the progress in another way which utilizes the connection between $\ket{\psi_x^{(t)}}$ and $\ket{\psi_x^{(t')}}$ for different $t$ and $t'$. Specifically, we show that any quantum query algorithm solving the quantum approximate counting problem has progress upper bounded by $O(\frac{t}{n})$ at time $t$, where $n$ is the number of items. The original adversarial method gives an $O(\frac{1}{\sqrt{n}})$ upper bound of the progress at any time $t$. Boyer et al.~\cite{boyer1998tight} gave a similar analysis of quantum search which utilizes the connection between states at different time $t$,  and got a tight lower bound of quantum search with a better constant factor compared to the hybrid argument. Since Reichardt \cite{reichardt2009span} proved that the generalized adversarial method is asymptotically tight, we cannot expect more by exploring connections between states at different time with identical query oracles. However, if each oracle can only be queried a limited number of times, our bound in \lem{counting-non-identical} is better than that obtained by the generalized adversarial method, since the progress bound $O(\frac{t}{n})$ is smaller in the early stages of the algorithm. We use this result to prove a query complexity lower bound of the quantum approximate counting problem with non-identical oracles. Since the quantum approximate counting problem can be reduced to estimating the mean of a Bernoulli random variable, we get a sample complexity lower bound of the quantum non-identical mean estimation problem in \thm{mean-estimation-lower-1-rough} for general $m$.

\subsection{Organization}
The rest of the paper is organized as follows. In \sec{prel} we formally define the input model and the quantum non-identical mean estimation problem, introduce the concept of parallel quantum query algorithms, and introduce quantum subroutines used in our algorithms. In \sec{upper-bound} we give quantum algorithms for estimating the mean of non-identically distributed bounded or sub-Gaussian random variables with quadratic speed-up. In \sec{lower-bound} we give two quantum query lower bounds of the quantum non-identical mean estimation problem based on reductions to low-depth quantum algorithms and the adversarial method with non-identical oracles, respectively.

\section{Preliminaries}
\label{sec:prel}
\subsection{Notations}
We denote \{1,2,\ldots,$n$\} by $[n]$. We use $|\psi\rangle_{A,B}$ to indicate that the state $|\psi\rangle$ is in quantum registers $A$ and $B$. For a quantum register $A$, we denote its number of qubits by $n_A$. For a boolean string $x\in \{0,1\}^n$, we denote its Hamming weight $|\{i\in [n]\mid x_i=1\}|$ by $|x|$. We abbreviate $|0^k\rangle$ as $|\mathbf{0}\rangle$ if $k$ can be inferred from the context.

\subsection{Input Model}

We first recall the definition of random variables and the input model of the classical mean estimation problem. 

\begin{definition}[Random variable]
  A finite random variable $X$ is a function $X\colon\Omega\to E$ for some probability space $(\Omega, p)$, where $\Omega$ is the finite sample space, $p$ is a probability measure on $\Omega$, and $E\subset \mathbb{R}$.
\end{definition}

Next, we assume that the random variable is the output of a quantum process $O_X$, and we can query $O_X$ as an oracle to access $X$. 

\begin{definition}[Quantum random variable]
  \label{def:query-access}
  For any finite random variable $X$, a quantum random variable encoding $X$ is a pair $(\mathcal{H},O_X)$, where $\mathcal{H}$ is a Hilbert space and $O_X$ is a unitary operator on $\mathcal{H}$ that performs the mapping \begin{align}
    O_X|\mathbf{0}\rangle=\sum_{x\in E}\sqrt{p(x)}|\psi_x\rangle|x\rangle
  \end{align}
  for some unknown garbage unit state $|\psi_x\rangle$. 
\end{definition}
Following the notation in \cite{hamoudi2021quantum}, we call each application to $U$ and $U^{\dagger}$ a \textit{quantum experiment}. We use the number of quantum experiments to measure the sample complexity of a quantum query algorithm.
\begin{definition}[Quantum experiment]
  Let $(\mathcal{H},O_X)$ be a quantum random variable. A quantum experiment is the process of applying $O_X$ or its inverse $O_X^{\dagger}$ or their controlled versions to a state in $\mathcal{H}$. 
\end{definition}
Performing a quantum experiment of a quantum random variable $(\mathcal{H},O_X)$ can be regarded as a query to the unitary oracle $O_X$ in the quantum query model, so the sample complexity is equivalent to the query complexity in this context, and we use the two terms interchangeably.

This input model is widely used in previous quantum mean estimation algorithms. The same oracle as defined in \define{query-access} is used in \cite{montanaro2015quantum}. Kothari and O'Donnell~\cite{kothari2023mean} used a similar input model except that they encode the probability distribution and the random variable mapping $\Omega\to \mathbb{R}$ in two oracles separately, and their algorithm also works well with the oracle in \define{query-access}. Hamoudi and Magniez~\cite{hamoudi2019quantum,hamoudi2021quantum} used a more general input model called ``q-random-variable'', where the value of the random variable is implicitly encoded in a register and can be compared with a constant or performed conditional Pauli rotations, and our oracle can be regarded as an instance of the ``q-random-variable''. Since the oracle in \de{query-access} already covers many common cases, we use it instead of the ``q-random-variable'' for simplicity and clarity. In fact, our quantum algorithm in \thm{mean-bounded} can also apply to the general ``q-random-variable''. 

The unitary $O_X$ is a quantum generalization of the process generating a sample of $X$. Bennett \cite{bennett1989time} proved that any classical algorithm using time $T$ and space $S$ can be modified to be a  reversible algorithm using time $O(T)$ and space $O(ST^{\epsilon})$ for any $\epsilon >0$, and hence can be simulated by a quantum circuit. Therefore, for any randomized algorithm $\mathcal{A}$, we can implement the oracle $O_X$ in \define{query-access} encoding the output distribution of $\mathcal{A}$ with a small overhead.

Another natural way for a quantum algorithm to access a random variable is to assume that several copies of $|\psi_X\rangle=\sum_{x\in E} \sqrt{p(x)}|x\rangle$ encoding the information of $X$ are given as the initial quantum state. This model is weaker than the one in \define{query-access} since it does not provide access to a unitary preparing $|\psi_X\rangle$. Hamoudi~\cite{hamoudi2021quantum} demonstrated that there is no quantum speed-up for the original mean estimation problem in this model. Therefore, it can be inferred that there is no quantum speed-up for the mean estimation problem of non-identically distributed random variables in this model, as it is a harder problem. 

Based on the definition of quantum random variable, we define the mean estimation problem of non-identically distributed random variables formally as the following task. 

\begin{task}[Quantum non-identical mean estimation]
  \label{task:mean-estimation}
  Let $(\mathcal{H},O_{X_1}),\ldots, (\mathcal{H},O_{X_T})$ be a sequence of quantum random variables on the same Hilbert space $\mathcal{H}$. Assume there exists $\mu$ and $\delta\in (0,1)$ such that each $\mu_i\coloneqq\mathbb{E}[X_i]$ satisfies $|\mu_i-\mu|\le \delta$ for all $i\in [T]$. Given the \textit{repetition parameter} $m \in \mathbb{N}$ and accuracy $\epsilon$ such that $\delta < c\epsilon$ for some constant $c<1$, the \textit{quantum non-identical mean estimation problem} is to estimate $\mu$ to within additive error $\epsilon$ with probability at least $2/3$ using each $O_{X_i}$ or $O_{X_i}^{\dagger}$ or their controlled versions at most $m$ times. 
\end{task}

The non-identity of quantum random variables means more than the non-identity of classical random variables. Specifically, the difference between two quantum random variables $(\mathcal{H},O_X),(\mathcal{H},O_Y)$ lies in the following three aspects: the results of applying $O_X$ and $O_Y$ to states orthogonal to $\ket{\mathbf{0}}$, the garbage state $|\psi_x\rangle$, and the random variables they encode. In contrast, the difference between two classical random variables is solely determined by the third aspect. Consequently, the quantum mean estimation problem of non-identically distributed random variables is more challenging than its classical counterpart.


\subsection{Parallel Quantum Query Algorithms}

The classical parallel algorithm implies that the algorithm can perform multiple operations simultaneously, which has become increasingly important in recent years with the development of multi-core processors. In the quantum setting, there is an additional reason to consider parallel algorithms: quantum states are fragile and susceptible to disruption by environmental factors, specifically decoherence. By reducing the computation time, parallel quantum algorithms can reduce the probability of decoherence. One example is parallel quantum query algorithms which can make multiple queries simultaneously, where a $p$-parallel query is defined as making $p$ parallel queries simultaneously. Zalka~\cite{zalka1999grover} gave an algorithm that makes $\sqrt{\frac{n}{p}}$ $p$-parallel queries to solve the unstructured search problem with $1$ marked item among $n$ items and showed that its query complexity is optimal. Subsequent works also analyzed the parallel quantum query complexity of quantum search~\cite{grover2004quantum}, quantum walk~\cite{jeffery2017optimal}, quantum counting~\cite{burchard2019lower}, and Hamiltonian simulation~\cite{zhang2024parallel}.

\subsection{Quantum Subroutines}
\begin{lemma}[Approximating unitary operators, Eq.~(4.63) of \cite{nielsen2001quantum}]\label{lem:approx}
Let $||\cdot||$ be the operator 2-norm.
For unitary operators $\{U_i\}_{i = 1}^m$, $\{V_i\}_{i = 1}^m$, it holds that
$$
\|U_mU_{m-1}\ldots U_1- V_mV_{m-1}\ldots V_1\| \leq \sum_{j = 1}^m\|U_j-V_j\|.
$$
\end{lemma}

\begin{lemma}[Amplitude estimation, Theorem 12 of \cite{brassard2002quantum}]\label{lem:amplitude_estimation}
  Given a unitary $U$ satisfying \begin{align}
    U\ket{\mathbf{0}}=\sqrt{p}\ket{\phi_1}\ket{1}+\sqrt{1-p}\ket{\phi_0}\ket{0}
  \end{align}
  for some $p\in[0,1]$, there exists a quantum circuit $C$ on a larger space such that the measurement outcome of $C\ket{\mathbf{0}}\ket{\mathbf{0}}$, $\tilde{p}$,  satisfies \begin{align}
    |\tilde{p}-p| \leq \frac{2 \pi \sqrt{p(1-p)}}{M}+\frac{\pi^2}{M^2}
  \end{align}
  with probability $\frac{8}{\pi^2}$, where $C$ has $M$ calls to the controlled versions of $I-2U\ket{\mathbf{0}}\bra{\mathbf{0}}U^{\dagger}$. Denote the algorithm by $\Amp(U,M)$.
\end{lemma}

\begin{lemma}[Fixed-point quantum search, \cite{yoder2014fixed}]{\label{lem:fixsearch}}
Let $A$ be a unitary and $\Pi$ be an orthogonal projector such that $\Pi A\ket{0} = \lambda \ket{\phi}$, where $\lambda \in \mathbb{R}$ and $|\phi\rangle$ is a normalized quantum state. There exists a quantum circuit $S_L = \Fix(A, \Pi, \epsilon)$ such that $||\ket{\phi} - S_L\ket{0} || \leq \epsilon$, consisting of $O(\log(1/\epsilon)/\lambda)$ queries to $A$, $A^{\dagger}$, and $\CPINOT$. Here $\CPINOT$ is the $\Pi$-controlled $\MYNOT$ operator
$$
\CPINOT = X\otimes \Pi + I \otimes (I-\Pi),
$$
where $X$ is the Pauli-X matrix.
\end{lemma}

\section{Upper Bound}\label{sec:upper-bound}
In this section, we first introduce an algorithm that solves \tsk{mean-estimation} for bounded random variables, and then generalize it to sub-Gaussian variables.
\subsection{Mean Estimation of Bounded Random Variables}
In this subsection, we introduce an algorithm that solves \tsk{mean-estimation} with quadratic speed-up given the condition that random variables $X_1,\ldots,X_T$ are bounded in $[L, H]$. According to the task, for each $i \in [T]$, oracle $O_{X_i}$ can be used at most $m$ times. 

For clarity, we describe the algorithm with two phases. Let \begin{align*}
\qquad \ket{\phi_i} = \frac{q_i}{\sqrt{2q_i^2-2q_i+1}} \ket{1}+\frac{1-q_i}{\sqrt{2q_i^2-2q_i+1}}\ket{0}.   
\end{align*} 
Here $q_i = \frac{\mu_i - L}{H-L} \in [0,1]$. For each $i \in [T]$, We will construct a quantum circuit $S_i$ that satisfies $S_i\ket{\mathbf{0}} \approx \ket{\phi_i}$ with $m$ calls to $O_{X_i}$. Then we will prove that performing amplitude estimation with these $S_i$ gives an $\epsilon$-additive estimation of $\mu$. 

\begin{theorem}
\label{thm:mean-bounded}
Assume that all random variables $X_1,\ldots, X_T$ in \tsk{mean-estimation} are bounded in $[L,H]$. Let $m$, $\epsilon$, $\delta$ in \alg{bounded} satisfy $m =\Omega(\log(\frac{H-L}{\epsilon}))$, $\epsilon=O\bigl(\frac{(\mu-L)(H-\mu)}{H-L}\bigr)
$, and $\delta < \epsilon/2$. \alg{bounded} solves this task if $T = \Omega(\frac{H-L}{\epsilon})$, using $O(\frac{H-L}{\epsilon}\log(\frac{H-L}{\epsilon}))$ quantum experiments in total.

\end{theorem}
\begin{algorithm}[tb]
\caption{Mean Estimation of Bounded Random Variables}\label{alg:bounded}
\begin{algorithmic}[1]
\STATE{\bfseries Input:} sequence of random variable oracle $\{O_{X_i}\}_{i = 1}^T$, accuracy $\epsilon$, mean difference $\delta$, repetition parameter $m$, lower bound $L$, upper bound $H$ 
\STATE{\bfseries Output:} mean estimation $\tilde{\mu}$ \\
// \ Construct quantum circuit $S_i$
\STATE Construct unitary $U_i$\\
\begin{equation*}
\begin{aligned}
    U_i:|\mathbf{0}\rangle|0\rangle & \xlongrightarrow{O_{X_i}\otimes I} \sum_{x\in E_i}\sqrt{p_i(x)}|\psi_x^{(i)}\rangle|x\rangle|0\rangle \\
    & \xlongrightarrow{\text{controlled rotation}} \sum_{x\in E_i}\sqrt{p_i(x)}|\psi_x^{(i)}\rangle|x\rangle\left(\sqrt{\frac{x-L}{H-L}}|1\rangle + \sqrt{\frac{H-x}{H-L}}|0\rangle\right) \\
\end{aligned}
\end{equation*}\label{lin:U_i}
\STATE Let $V_i = (U_i^{\dagger}\otimes I)(I\otimes \CNOT) (U_i\otimes I)$\label{lin:V_i}
\STATE Let $S_i = \Fix(V_i, \ket{\mathbf{0}}\ket{0}\bra{0}\bra{\mathbf{0}}\otimes I, \epsilon' = O(\epsilon^2/(H-L)^2) )$ \\
// \ Mean estimation using $S_i$ \label{lin:S_i}
\STATE Let $\tilde{p}$ be the output of $\Amp(S, M = O(\frac{H-L}{\epsilon}))$, where $S$ is arbitrarily replaced by $S_1,\ldots,S_T$.
\STATE Output $\tilde{\mu} =  \frac{\tilde{p}-\sqrt{\tilde{p}(1-\tilde{p})}}{2\tilde{p}-1}(H-L)+L$ 
\end{algorithmic}
\end{algorithm}

\begin{proof}
We first prove that $S_i$ in \lin{S_i} satisfies $S_i\ket{\mathbf{0}}\ket{0}\ket{0}=\sqrt{1-\epsilon_i}\ket{\mathbf{0}}\ket{0}\ket{\phi_i} + \sqrt{\epsilon_i}\ket{\mathrm{garbage}_i}$.  
According to the construction of $U_i$ in \lin{U_i} of \alg{bounded}, we have \begin{align}
    U_i|\mathbf{0}\rangle|0\rangle = \sqrt{q_i}|\psi_1^{(i)}\rangle|1\rangle + \sqrt{1-q_i}|\psi_0^{(i)}\rangle|0\rangle
  \end{align}
for some unit states $|\psi_1^{(i)}\rangle$ and $|\psi_0^{(i)}\rangle$. Consider the $V_i$ in \lin{V_i} where we append a qubit to the register. For any $b\in \{0,1\}$ we have 
\begin{align}
      \bra{b}\bra{0}\bra{\mathbf{0}} V_i \ket{\mathbf{0}}\ket{0}\ket{0} 
      & = ((U_i \otimes I) \ket{\mathbf{0}}\ket{0}\ket{b})^{\dagger}(I\otimes \CNOT) (U_i\otimes I)\ket{\mathbf{0}}\ket{0}\ket{0}\nonumber \\
      & = \Bigl(\sqrt{q_i}\bra{b}\bra{1}\bra{\psi_1^{(i)}} + \sqrt{1-q_i}\bra{b}\bra{0}\bra{\psi_0^{(i)}}\Bigr)\Bigl(\sqrt{q_i}\ket{\psi_1^{(i)}}\ket{1}\ket{1}+\sqrt{1-q_i}\ket{\psi_0^{(i)}}\ket{0}\ket{0}\Bigr)\nonumber \\
      & = \begin{cases} q_i & b = 1 \\ 1-q_i & b = 0,\end{cases}
  \end{align}
  which implies that
  \begin{align}
    V_i\ket{\mathbf{0}}\ket{0}\ket{0} 
    & = \sqrt{2q_i^2-2q_i+1}\ket{\mathbf{0}}\ket{0}\biggl(\frac{q_i}{\sqrt{2q_i^2-2q_i+1}}\ket{1}+\frac{1-q_i}{\sqrt{2q_i^2-2q_i+1}}\ket{0}\biggr)\nonumber\\
    &\qquad\qquad\qquad+ \sqrt{2q_i-2q_i^2}\ket{\mathrm{garbage}_i},
  \end{align}
  where $\ket{\mathrm{garbage}_i}$ is a unit garbage state and $(I\otimes\bra{0}\bra{\mathbf{0}})|\mathrm{garbage}_i\rangle = 0$. Moreover, we define 
  \begin{align}
      \ket{\phi_i} =\frac{q_i}{\sqrt{2q_i^2-2q_i+1}} \ket{1}+\frac{1-q_i}{\sqrt{2q_i^2-2q_i+1}}\ket{0}, \qquad \ket{s_i} = V_i\ket{\mathbf{0}}\ket{0}\ket{0}.
  \end{align}
  Under these notations, we have 
  \begin{align}
\left(\ket{\mathbf{0}}\ket{0}\bra{0}\bra{\mathbf{0}}\otimes I\right)V_i\ket{\mathbf{0}}\ket{0}\ket{0} =  \sqrt{2q_i^2-2q_i+1} \ket{\mathbf{0}}\ket{0}\ket{\phi_i} .
\end{align}
Together with \lem{fixsearch} and the fact that $\sqrt{2q_i^2-2q_i+1} \geq \frac{1}{\sqrt{2}}$, we know that $S_i$ in \lin{S_i} satisfies
\begin{align}
  \label{eq:Si}
  S_i\ket{\mathbf{0}}\ket{0}\ket{0}=\sqrt{1-\epsilon_i}\ket{\mathbf{0}}\ket{0}\ket{\phi_i} + \sqrt{\epsilon_i}\ket{\mathrm{garbage}_i},
\end{align}
where $\epsilon_i \leq \epsilon'$ and $S_i$ contains $O\big(\log\frac{1}{\epsilon'}\big)=O\big(\log(\frac{H-L}{\epsilon})\big)$ calls to $V_i$.

 Let 
$$
q=\frac{\mu-L}{H-L}\in [0,1], \qquad \ket{\phi} = \frac{q}{\sqrt{2q^2-2q+1}} \ket{1}+\frac{1-q}{\sqrt{2q^2-2q+1}}\ket{0},
$$
and $S$ be a unitary such that
\begin{align}
  \label{eq:S}
  S\ket{\mathbf{0}}\ket{0}\ket{0} = \ket{\mathbf{0}}\ket{0}\ket{\phi}.
\end{align}
Performing an amplitude estimation using $\{S_i\}_{i = 1}^T$ provides a result similar to an amplitude estimation using $S$, and thus provides a mean estimation with additive error $O(\epsilon)$. See the details in \append{mean-bounded}

Each $V_i$ uses two quantum experiments, each $S_i$ uses $O(\log(\frac{H-L}{\epsilon}))$ calls to $V_i$, and $C'$ uses $M=O(\frac{H-L}{\epsilon})$ calls to controlled $S_i$. Therefore, the total number of quantum experiments is $O\big(\frac{H-L}{\epsilon}\log(\frac{H-L}{\epsilon})\big)$.
\end{proof}

\begin{remark}\label{rem:delta}
For every $i\in [T]$, $S_i$ can be seen as an approximation of unitary $S$. The slight difference $\delta$ among different $\mu_i$ only causes a part of approximation error which is bounded by $\epsilon$. Therefore, this difference is tolerable in our algorithm. See \eqn{delta_1} and \eqn{delta_2} for more details.
\end{remark}

\subsection{Mean Estimation of Sub-Gaussian Random Variables}
In this subsection, we consider the quantum non-identical mean estimation problem of sub-Gaussian random variables. 
\begin{definition}
    \label{def:sub-gaussian}
  A random variable $X$ is sub-Gaussian with parameter $K$ if for all $t\ge 0$ \begin{align}
    \label{eq:def-sub-gaussian}
    \mathbb{P}[|X-\mathbb{E}[X]|\ge t] \le 2\exp\Bigl(-\frac{t^2}{2K^2}\Bigr).
  \end{align}
\end{definition}

We first give a quantum algorithm estimating the mean of non-identically distributed sub-Gaussian random variables with quadratic speed-up if the mean of the random variables are bounded by their sub-Gaussian parameter. This case can be reduced to the case of bounded random variables by truncation. Then, we show that this algorithm can be generalized to any sub-Gaussian random variable.
\begin{lemma}
\label{lem:mean-sub-gaussian}
Suppose all random variables $X_1,\ldots, X_T$ in \tsk{mean-estimation} are sub-Gaussian with parameter $K$ and their mean satisfies $|\mu_i| \le R$, $R \leq K$. Let $m,R,K,\epsilon$, $\delta$ in \alg{bounded_Gaussian} satisfies that $m =\Omega\Bigl(\log\Big(\frac{K\sqrt{\log(\frac{K}{\epsilon})}}{\epsilon}\Big)\Bigr)$, $\epsilon=O(K)$, and $\delta < \epsilon/4$. \alg{bounded_Gaussian} solves  \tsk{mean-estimation} if $T = \Omega(\frac{K\sqrt{\log (\frac{K}{\epsilon})}}{\epsilon})$, using $O\Bigl(\frac{K\sqrt{\log(\frac{K}{\epsilon})}}{\epsilon}\log\Big(\frac{K\sqrt{\log(\frac{K}{\epsilon})}}{\epsilon}\Big)\Bigr)$ quantum experiments in total.
\end{lemma}

\begin{algorithm}[tb]
\caption{Mean Estimation of Mean-Bounded sub-Gaussian Random Variable}\label{alg:bounded_Gaussian}
\begin{algorithmic}[1]
\STATE{\bfseries Input:} sequence of random variable oracle $\{O_{X_i}\}_{i = 1}^T$, accuracy $\epsilon$, mean difference $\delta$, repetition parameter $m$, upper bound for mean $R$, sub-Gaussian parameter $K$
\STATE{\bfseries Output:} mean estimation $\tilde{\mu}$
\STATE  Let $\Delta =K\max\Big\{\sqrt{4\log(\frac{128K}{\epsilon})},\sqrt{2\log(\frac{32R}{\epsilon})}\Big\}$, $L=-R-\Delta$, $H=R+\Delta$
\STATE Construct unitary $O_{\tilde{X}_i}$\\
\begin{equation*}
\begin{aligned}
    O_{\tilde{X}_i}:|\mathbf{0}\rangle\ket{\mathbf{0}} & \xlongrightarrow{O_{X_i}\otimes I} \sum_{x\in E_i}\sqrt{p_i(x)}|\psi_x^{(i)}\rangle|x\rangle\ket{\mathbf{0}} \\
    & \xlongrightarrow{\CNOT} \sum_{x\in [L,H]}\sqrt{p_i(x)}|\psi_x^{(i)}\ket{x}\ket{x}+\sum_{x \in E_i\setminus [L,H]}\sqrt{p_i(x)}|\psi_x^{(i)}\ket{x}\ket{\mathbf{0}}
\end{aligned}
\end{equation*}
\STATE Output $\tilde{\mu} = $\alg{bounded}($\{O_{\tilde{X}_i}\}_{i = 1}^T$, accuracy $\epsilon$, mean difference $\delta = \epsilon/2$, $m$, $L$, $H$)
\end{algorithmic}
\end{algorithm}

Quantum random variable $\tilde{X}_i$ generated by oracle $O_{\tilde{X}_i}$ in \alg{bounded_Gaussian} is a truncated version of $X_i$. Calculation shows that the mean difference is within $\frac{\epsilon}{2}$, thus \alg{bounded_Gaussian} provides an estimation with $O(\epsilon)$ additive error.

\begin{proof}
See \append{mean-sub-gaussian}.
\end{proof}

For general sub-Gaussian distributions, we first use $O(1)$ classical samples to estimate the mean of these sub-Gaussian random variables up to additive error $K/2$, and then shift the random variables by subtracting the approximate mean so that the shifted random variables have mean bounded by their sub-Gaussian parameter. After that, we can use \lem{mean-sub-gaussian} to estimate the mean of the shifted random variables.  

\begin{theorem}
\label{thm:mean-sub-gaussian}
Assume all random variables $X_1,\ldots, X_T$ in \tsk{mean-estimation} are sub-Gaussian with parameter $K$. Let $m,K,\delta,\epsilon$ in \alg{Gaussian} satisfy that $m =\Omega\Bigl(\log\Big(\frac{K\sqrt{\log(\frac{K}{\epsilon})}}{\epsilon}\Big)\Bigr)$, $\epsilon=O(K)$, and $\delta < \epsilon/4$. \alg{Gaussian} solves \tsk{mean-estimation} if $T = \Omega(\frac{K\sqrt{\log(\frac{K}{\epsilon})}}{\epsilon})$, using $O\Bigl(\frac{K\sqrt{\log(\frac{K}{\epsilon})}}{\epsilon}\log\Big(\frac{K\sqrt{\log(\frac{K}{\epsilon})}}{\epsilon}\Big)\Bigr)$ quantum experiments in total.
\end{theorem}

\begin{algorithm}[tb]
\caption{Mean Estimation of sub-Gaussian Random Variable}\label{alg:Gaussian}
\begin{algorithmic}[1]
\STATE{\bfseries Input:} sequence of random variable oracle $\{O_{X_i}\}_{i = 1}^T$, accuracy $\epsilon$, repetition parameter $m$, sub-Gaussian parameter $K$
\STATE{\bfseries Output:} mean estimation $\tilde{\mu}$
\STATE  Perform $N=\lceil8\log(20)\rceil$ times classical experiments on arbitrary $X_i$ and let the average of the samples be $\hat{\mu}$\label{lin:classical}
\STATE Construct unitary $O_{X_i'}$\\
\begin{equation*}
\begin{aligned}
    O_{X_i'}:|\mathbf{0}\rangle\ket{\mathbf{0}} & \xlongrightarrow{O_{X_i}\otimes I} \sum_{x\in E_i }\sqrt{p_i(x)}|\psi_x^{(i)}|x\rangle\ket{\mathbf{0}} \\
    & \longrightarrow \sum_{x\in E_i}\sqrt{p_i(x)}\ket{\psi_x^{(i)}|x}\ket{x - \hat{\mu}}
\end{aligned}
\end{equation*}
\STATE Output $\tilde{\mu} = $\alg{bounded_Gaussian}($\{O_{X_i'}\}_{i = 1}^T$, accuracy $\epsilon$, mean difference $\delta=\epsilon/4$, $m$, upper bound for mean $R = K$, sub-Gaussian parameter $K$)
\end{algorithmic}
\end{algorithm}

\begin{proof}
Classical experiment in \lin{classical} can be naturally implemented by quantum access to random variable. For any $i\in [T]$, by applying $O_{X_i}$ to $\ket{\mathbf{0}}$ and measuring the second register in computational basis, we can get a classical sample of $X_i$. Since $\hat{\mu}$ is the average value of $N=\lceil8\log(20)\rceil$ samples, by the Hoeffding inequality for sub-Gaussian distributions \cite{vershynin2018high}, we have \begin{align}
    \mathbb{P}[|\hat{\mu}-\mathbb{E}[\hat{\mu}]|\ge \frac{K}{2}] \le 2\exp\Bigl(-\frac{N}{2K^2}\frac{K^2}{4}\Bigr) \le \frac{1}{10}.
  \end{align}
In addition, since $|\mu_i-\mu|\le \delta$ for all $i\in [T]$, we have \begin{align}
    |\mathbb{E}[\hat{\mu}]-\mu| \le \delta.
  \end{align}
$O_{X_i}'$ can be seen as quantum query to random variable $X_i' = X_i-\hat{\mu}$. With probability at least $\frac{9}{10}$, we have \begin{align}
    |\mathbb{E}[X_i']| = |\mathbb{E}[X_i]-\hat{\mu}| \le |\mathbb{E}[X_i]-\E[\hat{\mu}]|+|\hat{\mu}-\mathbb{E}[\hat{\mu}]|\le \delta + \frac{K}{2}\le K.
  \end{align}
Therefore, by \lem{mean-sub-gaussian} with $R=K$, $m=\Omega\Bigl(\log\Big(\frac{K\sqrt{\log(\frac{K}{\epsilon})}}{\epsilon}\Big)\Bigr)$ and $X_i'=X_i-\hat{\mu}$, we can estimate $\mu-\hat{\mu}$ with additive error $O(\epsilon)$ with probability at least $\frac{4}{5}$ using $O\Bigl(\frac{K\sqrt{\log(\frac{K}{\epsilon})}}{\epsilon}\log\Big(\frac{K\sqrt{\log(\frac{K}{\epsilon})}}{\epsilon}\Big)\Bigr)$ quantum experiments. Subtracting $\hat{\mu}$ from the estimate gives the final output of the algorithm which is an $\epsilon$-additive estimate of $\mu$ with probability at least $\frac{4}{5}\cdot\frac{9}{10}\ge \frac{2}{3}$. 
\end{proof}

\section{Lower Bound}
\label{sec:lower-bound}

In this section, we prove sample complexity lower bounds for the quantum non-identical mean estimation problem in \tsk{mean-estimation}.

Let $m$ be the repetition parameter defined \tsk{mean-estimation}. In \sec{lower-bound-sg}, we give a sample complexity lower bound for $m=1$, and show there is no quantum speed-up compared to classical algorithms. In \sec{lower-bound-po}, we give a sample complexity lower bound for $m\ge 1$.

\subsection{Lower Bound for $m=1$}
\label{sec:lower-bound-sg}

Let $X$ be a finite random variable with support $E$. Let $(\mathcal{H},O_X)$ be a quantum random variable in \define{query-access}, i.e.,
\begin{align}
    \label{eqn:def-query}
    O_X|\mathbf{0}\rangle=\sum_{x\in E}\sqrt{p(x)}|\psi_x\rangle|x\rangle,
  \end{align}
and we denote the output state by $|\psi_X\rangle$. A $p$-parallel query to $O_X$ is to apply the unitary $O_X^{\otimes q}$ or $O_X^{\dagger\otimes q}$ for $q\le p$.

Note that Eq.~\eqn{def-query} only restricts the outcome of applying $O_X$ on $\ket{\mathbf{0}}$, so the quantum random variable encoding the same $X$ can be different. Throughout \sec{lower-bound-sg}, we assume all quantum random variables encode the same finite random variable $X$. Given that $m=1$, the algorithm can perform only one quantum experiment for each quantum random variable.

We use the quantum query model to analyze the sample complexity of the quantum non-identical mean estimation since every quantum experiment can be regarded as a query to the oracle $O_X$. A $T$-query quantum algorithm starts from an all-0 state $\ket{\mathbf{0}}_Q\ket{\mathbf{0}}_W$, and then interleaves fixed unitary operations $U_0, U_1,\ldots, U_T$ with queries. Suppose different oracles are queried at different time, and we denote the $t$-th oracle queried by the algorithm as $O_X^{(t)}$. Without loss of generality, we assume that all queries are applied to register $\ket{\mathbf{0}}_Q$ and $U_0, U_1,\ldots, U_T$ are applied to $\ket{\mathbf{0}}_Q\ket{\mathbf{0}}_W$. Whether to apply $O_X^{(t)}$ or $(O_X^{(t)})^{\dagger}$ needs to be determined in advance, and the choices can be represented by $T$ boolean variables $a_1,\ldots,a_T\in \{-1,1\}$ such that
\begin{align}
    (O_X^{(t)})^{a_t} = \begin{cases}
        O_X^{(t)} & \text{if }a_t=1,\\
        (O_X^{(t)})^{\dagger} & \text{if } a_t=-1.
    \end{cases}
\end{align}

For any $ 1\le t\le T$, let \begin{align}
    |\psi^{(t)}\rangle:=U_t(O_X^{(t)})^{a_t}\cdots (O_X^{(1)})^{a_1} U_0\ket{\mathbf{0}}_Q\ket{\mathbf{0}}_W.
\end{align}
Hence the final state of the algorithm is $|\psi^{(T)}\rangle$.

At the end of the algorithm, we will measure $|\psi^{(T)}\rangle$ and let the projection onto the correct outputs be $\Pi_{c}$, and the success probability of the algorithm is hence
\begin{align}
    \|\Pi_c |\psi^{(T)}\rangle\|^2.
\end{align}

\subsubsection{Reduction to Low-depth Quantum Algorithms}

For a quantum circuit with oracles, the query depth is the maximum number of queries on any path from an input qubit to an output qubit. In this section, we prove that the behavior of a quantum algorithm querying $T$ non-identical oracles can be simulated by a low query depth quantum algorithm with the same number of queries. Actually, we will show that the behavior of the algorithm can be simulated by a quantum circuit using two $T$-parallel queries.

For any $1\le t\le T$, let \begin{align}
    \label{eqn:def-beg}
    |\phi_{\mathrm{beg}}^{(t)}\rangle&\coloneqq\begin{cases}
        \ket{\mathbf{0}}& \text{if } a_t = 1,\\
        |\psi_X\rangle & \text{if } a_t=-1,
    \end{cases} \\
    \label{eqn:def-end}
    |\phi_{\mathrm{end}}^{(t)}\rangle&\coloneqq\begin{cases}
        |\psi_X\rangle& \text{if } a_t = 1,\\
        \ket{\mathbf{0}} & \text{if } a_t=-1,
    \end{cases}
    \end{align} 
 so that \begin{align}
     (O_X^{(t)})^{a_t}|\phi_{\mathrm{beg}}^{(t)}\rangle = |\phi_{\mathrm{end}}^{(t)}\rangle.
 \end{align}
 This is the only subspace that $(O_X^{(t)})^{a_t}$'s behavior is fixed and defined by Eq.~\eqn{def-query}. 

For any $1\le t\le T$, let \begin{align}
    \Pi_{\mathrm{beg}}^{(t)}:= |\phi_{\mathrm{beg}}^{(t)}\rangle\langle\phi_{\mathrm{beg}}^{(t)}|\otimes I,
\end{align}
and
\begin{align}
\label{eqn:def-eff-state}
     |\psi_{\mathrm{eff}}^{(t)}\rangle_{Q,W} :=&\ (O_X^{(t)})^{a_t}\Pi_{\mathrm{beg}}^{(t)}U_{t-1}(O_X^{(t-1)})^{a_{t-1}}\Pi_{\mathrm{beg}}^{(t-1)} \cdots  U_1(O_X^{(1)})^{a_1}\Pi_{\mathrm{beg}}^{(1)}U_0\ket{\mathbf{0}}_Q\ket{\mathbf{0}}_W\\
     =&\ (|\phi_{\mathrm{end}}^{(t)}\rangle\langle\phi_{\mathrm{beg}}^{(t)}|\otimes I)U_{t-1}\cdots  U_1(|\phi_{\mathrm{end}}^{(1)}\rangle\langle\phi_{\mathrm{beg}}^{(1)}|\otimes I)U_0\ket{\mathbf{0}}_Q\ket{\mathbf{0}}_W.
\end{align}
These states are fixed no matter what the queries $O_X^{(t)}$ are, since all queries in Eq.~\eqn{def-eff-state} are applied to the subspace that its behavior is defined by Eq.~\eqn{def-query}.

We show in the following lemma that $|\psi_{\mathrm{eff}}^{(t)}\rangle$ can be prepared by a quantum algorithm using two $t$-parallel queries after post-selection.

\begin{lemma}
    \label{lem:eff-state}
    Given a $T$-query quantum algorithm acting on registers $Q$ and $W$, for any $0\le t\le T$, $|\psi_{\mathrm{eff}}^{(t)}\rangle$ defined in Eq.~\eqn{def-eff-state} can be prepared by another quantum circuit $V_t^{\mathrm{low}}$ using two $t$-parallel queries to any unitary oracle $O_X$ satisfying Eq.~\eqn{def-query} after post-selection, namely, \begin{align}
        \label{eqn:def-V-low}
        \bigl(I_{W,Q_t}\otimes \langle\mathbf{0}|_{Q_0,\ldots,Q_{t-1}}\bigr)V_t^{\mathrm{low}}\ket{\mathbf{0}}_{W,Q_0,\ldots,Q_t},
        \end{align}  
    where $Q_0,\ldots,Q_t$ are $t+1$ registers with $n_Q$ qubits.
\end{lemma}
\begin{proof}
    For all $1\le t< T$, from the definition of $|\psi_{\mathrm{eff}}^{(t)}\rangle$, it can be written as \begin{align}
     \label{eqn:eff-pro}
        |\psi_{\mathrm{eff}}^{(t)}\rangle = |\phi_{\mathrm{end}}^{(t)}\rangle|\phi_{\mathrm{W}}^{(t)}\rangle
    \end{align}
    for some unnormalized state $|\phi_{\mathrm{W}}^{(t)}\rangle$, then we have \begin{align}
    |\phi_{\mathrm{end}}^{(t+1)}\rangle|\phi_{\mathrm{W}}^{(t+1)}\rangle=|\psi_{\mathrm{eff}}^{(t+1)}\rangle = (|\phi_{\mathrm{end}}^{(t+1)}\rangle\langle\phi_{\mathrm{beg}}^{(t+1)}|\otimes I)U_{t}|\phi_{\mathrm{end}}^{(t)}\rangle|\phi_{\mathrm{W}}^{(t)}\rangle.
    \end{align}
    Apply $\langle \phi_{\mathrm{end}}^{(t+1)}|\otimes I$ to both sides we have 
    \begin{align}
    \label{eqn:eff-recursion}
        |\phi_{\mathrm{W}}^{(t+1)}\rangle = (\langle\phi_{\mathrm{beg}}^{(t+1)}|\otimes I)U_{t}|\phi_{\mathrm{end}}^{(t)}\rangle|\phi_{\mathrm{W}}^{(t)}\rangle.
    \end{align}
    Define \begin{align}
        \ket{\psi_{\mathrm{eff}}^{(0)}}=\ket{\mathbf{0}}\ket{\mathbf{0}}, \quad \ket{\phi_{\mathrm{end}}^{(0)}}=\ket{\mathbf{0}},\quad \ket{\phi_{\mathrm{W}}^{(0)}}=\ket{\mathbf{0}},
    \end{align}
    so that Eq.~\eqn{eff-pro} and Eq.~\eqn{eff-recursion} also hold for $t=0$.

    To construct the required circuit, We prove the following stronger statement.
    \begin{statement}\label{sta:tement}
        Let $O_X$ be any unitary satisfying Eq.~\eqn{def-query}, and $U_0^{\mathrm{low}},\ldots,U_T^{\mathrm{low}}$ be a sequence of quantum circuits satisfying $U_{0}^{\mathrm{low}}=I$ and \begin{align}
            \label{eqn:recursion-U-low}
                U_{t+1}^{\mathrm{low}} = \begin{cases}
                    ((U_t)_{Q_t,W} \otimes I) \cdot (U_{t}^{\mathrm{low}}\otimes (O_X)_{Q_{t+1}}) & \text{if } a_{t+1}=1, \\
                    ((U_t)_{Q_t,W} \otimes I) \cdot (U_{t}^{\mathrm{low}}\otimes I_{Q_{t+1}}) & \text{if } a_{t+1}=-1,
                \end{cases}
            \end{align}
            for all $0\le t < T$. The quantum circuit $U_{t}^{\mathrm{low}}$ can prepare $|\psi_{\mathrm{eff}}^{(t)}\rangle$ after post-selection, namely, \begin{align}
        \label{eqn:def-U-low}
    |psi_{\mathrm{eff}}^{(t)}\rangle = \Bigl(I_{W,Q_t}\bigotimes_{i=1}^{t} \langle\phi_{\mathrm{beg}}^{(i)}|_{Q_{i-1}}\Bigr)U_t^{\mathrm{low}}\ket{\mathbf{0}}_{W,Q_0,\ldots,Q_t},
    \end{align} 
    for any $0\le t\le T$.
    \end{statement}\begin{proof}
    See \append{statement}.
        \end{proof}
    
    The number of queries in $U_t^{\mathrm{low}}$ is $|\{a_i=1\mid i\in [t]\}|$.
    Let \begin{align}
        V_t^{\mathrm{low}} = \bigotimes_{1\le i\le t,a_i=-1}(O_X^{\dagger})_{Q_i}U_{t}^{\mathrm{low}},
    \end{align}
    then from Eq.~\eqn{def-U-low} we have \begin{align}
    \bigl(I_{W,Q_t}\otimes \langle\mathbf{0}|_{Q_0,\ldots,Q_{t-1}}\bigr)V_t^{\mathrm{low}}\ket{\mathbf{0}}_{W,Q_0,\ldots,Q_t},
    \end{align} 
    for all $0\le t\le T$.

    The number of queries in $V_t^{\mathrm{low}}$ is \begin{align}
        |\{a_i=1\mid i\in [t]\}|+|\{a_i=-1\mid i\in [t]\}|=t.
    \end{align}
    Conditioning on the state in registers $Q_0,\ldots,Q_{t-1}$ to be $\ket{\mathbf{0}}$, $V_t^{\mathrm{low}}$ prepares $|\psi_{\mathrm{eff}}^{(t)}\rangle_{Q_t,W}$ and  uses two $t$-parallel queries.
\end{proof}

Next, we demonstrate that $U_T|\psi_{\mathrm{eff}}^{(T)}\rangle$ is the only useful component in the final state $|\psi^{(T)}\rangle$,
since other parts can be controlled by $O_X^{(t)}$ to make the result worse. Before that, we prove the following useful lemma. 
\begin{lemma}
\label{lem:can-control}
     For any $T$-query quantum algorithm acting on registers $Q$, $W$, and any finite random variable $X$ on $(\Omega,p)$, if $\dim{\mathcal{H}_Q}>2\dim{\mathcal{H}_W}$, there exists a sequence of quantum random variables  $(\mathcal{H}_Q,O_X^{(1)}), \ldots, (\mathcal{H}_Q,O_X^{(T-1)})$ such that for any $0 \le t < T$ \begin{align}
        \label{eqn:recur-control}
         |\psi^{(t)}\rangle = |\phi_{\mathrm{beg}}^{(t+1)}\rangle|\phi_W^{(t+1)}\rangle + |\psi^{(t)}_{\perp}\rangle,
     \end{align}
     for some unnormalized state $|\psi^{(t)}_{\perp}\rangle$ orthogonal to $|\phi_{\mathrm{beg}}^{(t+1)}\rangle\otimes\mathcal{H}_W$.
\end{lemma}
\begin{proof}
By induction. See the details in \append{can-control}.
\end{proof}

Now we prove that $U_T|\psi_{\mathrm{eff}}^{(T)}\rangle$ is the only useful component in the final state $|\psi^{(T)}\rangle$.

\begin{lemma}
    \label{lem:useful-component}
    Suppose that $X$ is a finite random variable. For any $T$-query quantum algorithm acting on registers $Q$, $W$, and any projection $\Pi_c$, if $\dim{\mathcal{H}_Q}>2\dim{\mathcal{H}_W}$ and $\dim{\mathcal{H}_Q}\ge2\dim{\mathrm{Im}(\Pi_c)}$, then there exists a sequence of quantum random variables $(\mathcal{H}_Q,O_X^{(1)}), \ldots, (\mathcal{H}_Q,O_X^{(T)})$ such that \begin{align}
        \label{eq:useful-component}
        \|\Pi_c |\psi^{(T)}\rangle\|^2 = \|\Pi_c U_T |\psi_{\mathrm{eff}}^{(T)}\rangle\|^2.
    \end{align}
\end{lemma}
\begin{proof}
    Note that 
    \begin{align}
        |\psi^{(T)}\rangle &= U_T(O_X^{(T)})^{a_T}|\psi^{T-1}\rangle \\
        &=U_T(O_X^{(T)})^{a_T}(|\phi_{\mathrm{beg}}^{(T)}\rangle|\phi_W^{(T)}\rangle+|\psi_{\perp}^{(T-1)}\rangle)\\
        &=U_T|\phi_{\mathrm{end}}^{(T)}\rangle|\phi_W^{(T)}\rangle + U_T(O_X^{(T)})^{a_T}|\psi_{\perp}^{(T-1)}\rangle\\
        &=U_T|\psi_{\mathrm{eff}}^{(T)}\rangle + U_T(O_X^{(T)})^{a_T}|\psi_{\perp}^{(T-1)}\rangle.
    \end{align}
    To satisfy Eq.~\eq{useful-component}, we need to find a unitary operator $O_X^{(T)}$ such that \begin{align}
        \Pi_c U_T(O_X^{(T)})^{a_T}|\psi_{\perp}^{(T-1)}\rangle=0,
    \end{align} which means \begin{align}
        \label{eqn:useful-flag}
        (O_X^{(T)})^{a_{T}}|\psi^{(T-1)}_{\perp}\rangle\in (U_{T}^{\dagger}\mathrm{Im}(\Pi_c))^{\perp}.
    \end{align}
    Note that Eq.~\eqn{useful-flag} has a similar form as Eq.~\eqn{success-flag}, so we can construct $O_X^{(T)}$ in the same way as we construct $O_X^{(t+1)}$ in the proof of \lem{can-control}. By similar argument to \lem{can-control}, we can prove that if 
    \begin{align}
        \dim{\mathcal{H}_Q}>\dim{\mathcal{H}_W}+\dim{\mathrm{Im}(\Pi_c)},
    \end{align}
    there exists $O_X^{(T)}$ such that Eq.~\eq{useful-component} holds. By assumptions that $\dim{\mathcal{H}_Q}>2\dim{\mathcal{H}_W}$ and $\dim{\mathcal{H}_Q}\ge2\dim{\mathrm{Im}(\Pi_c)}$, we can conclude that Eq.~\eq{useful-component} holds.
\end{proof}

In conclusion, there exists a sequence of quantum random variables such that the output of a $T$-query quantum algorithm can be simulated by a quantum algorithm using two $T$-parallel queries.
\begin{theorem}
    \label{thm:low-depth}
    For any $T$-query quantum algorithm $\mathcal{A}$ acting on registers $Q$, $W$, and any projection $\Pi_c$, suppose that $\dim{\mathcal{H}_Q}>2\dim{\mathcal{H}_W}$ and $\dim{\mathcal{H}_Q}\ge2\dim{\mathrm{Im}(\Pi_c)}$. Let $\ket{\psi^{(T)}}$ be the final state of the algorithm. There exists another quantum circuit $U^{\mathrm{low}}$ using two $T$-parallel queries such that for any finite random variable $X$, there is a sequence of quantum random variables $(\mathcal{H}_Q,O_X^{(1)}), \ldots, (\mathcal{H}_Q,O_X^{(T)})$ satisfying
    \begin{align}
        \|\Pi_c\ket{\psi^{(T)}}\|^2=\|\bigl(\Pi_c\otimes \langle\mathbf{0}|_{Q_0,\ldots,Q_{T-1}}\bigr)U^{\mathrm{low}}\ket{\mathbf{0}}_{W,Q_0,\ldots,Q_T}\|^2,
        \end{align}  
    where $Q_0,\ldots,Q_T$ are $T+1$ registers with $n_Q$ qubits.
\end{theorem}
\begin{proof}
    Let $V_{T}^{\mathrm{low}}$ be the low-depth quantum circuit defined in \lem{eff-state}, and $U_T$ be the unitary in algorithm $\mathcal{A}$ at time step $T$. By \lem{eff-state}, the unitary $U^{\mathrm{low}}=((U_T)_{Q_T,W}\otimes I)V_{T}^{\mathrm{low}}$ satisfies \begin{align}
        \bigl(I\otimes \langle\mathbf{0}|_{Q_0,\ldots,Q_{T-1}}\bigr)U^{\mathrm{low}}\ket{\mathbf{0}}_{W,Q_0,\ldots,Q_T}=U_T|\psi_{\mathrm{eff}}^{(T)}\rangle_{Q_T,W}.
    \end{align}
    By \lem{useful-component}, there exists a sequence of quantum random variables $(\mathcal{H}_Q,O_X^{(1)}), \ldots, (\mathcal{H}_Q,O_X^{(T)})$ such that \begin{align}
        \|\Pi_c\ket{\psi^{(T)}}\|^2=\|\Pi_c U_T |\psi_{\mathrm{eff}}^{(T)}\rangle\|^2=\|\bigl(\Pi_c\otimes \langle\mathbf{0}|_{Q_0,\ldots,Q_{T-1}}\bigr)U^{\mathrm{low}}\ket{\mathbf{0}}_{W,Q_0,\ldots,Q_T}\|^2.
    \end{align}
\end{proof}

\subsubsection{Lower Bounds for Low-depth Quantum Mean Estimation Algorithms}

Given an input $x = x_0\ldots x_{n-1}\in \{0,1\}^n$, the quantum query to it is a unitary $O_x$ such that \begin{align}
\label{eq:flip-oracle}
    O_x|i\rangle|b\rangle = |i\rangle|b\oplus x_i\rangle
\end{align}
for all $i\in [n]$ and $b\in \{0,1\}$.

The \textit{approximate counting} problem is that given $O_x$, output an estimate of $|x|$ up to error $\epsilon$ with high probability. From another perspective, we can think of $[n]$ as a sample space $\Omega$ with uniform distribution $P$, and $X\colon\Omega\to \{0,1\}$ is a Bernoulli random variable such that $X(i) = x_i$, and the mean of $X$ is \begin{align}
    p=\frac{|x|}{n}.
\end{align}
Note that \begin{align}
    \label{eqn:flip-to-state}
    |0\rangle|0\rangle \xrightarrow{\text{Hardmard gates}} \sum_{i=1}^n \frac{1}{\sqrt{n}}|i\rangle|0\rangle\xrightarrow{I\otimes O_x}\sum_{i=1}^n \frac{1}{\sqrt{n}}|i\rangle|X(i)\rangle,
\end{align}
which means we can implement the oracle to $X$ with one query to $O_x$. Hence, the approximate counting problem can be reduced to the mean estimation problem. 

A $k$-parallel query call to $x$ is\begin{align}
    O_x^{\otimes k}\left|i_1, \ldots, i_k , b_1, \ldots, b_k\right\rangle=\left|i_1, \ldots, i_k , b_1 \oplus x_{i_1}, \ldots, b_k \oplus x_{i_k}\right\rangle
\end{align}

\cite{burchard2019lower} proved a $k$-parallel query lower bound of the approximate counting problem.

\begin{theorem}[{\cite{burchard2019lower}}]
    \label{thm:k-parallel-lower}
    For any quantum query algorithm and boolean string $x\in \{0,1\}^{n}$, 
    \begin{align}
         \Omega\Biggl(\frac{\binom{n-|x|}{\epsilon n}\binom{|x|+\epsilon n}{|x|}}{k\binom{n-|x|-1}{\epsilon n-1}\binom{|x|+\epsilon n-1}{|x|}} \Biggr)=\Omega\Bigl(\frac{p(1-p)}{\epsilon^2 k}\Bigr)
    \end{align}
    $k$-parallel queries to $O_x$ is necessary to estimate $p=\frac{|x|}{n}$ to within additive error $\epsilon$.
\end{theorem} 

By \thm{k-parallel-lower}, if we want to use constant $k$-parallel queries to estimate $p$ up to additive error $\epsilon$, $k$ needs to satisfy \begin{align}
    \frac{p(1-p)}{\epsilon^2 k} = O(1),
\end{align}
which means \begin{align}
    k = \Omega\Bigl(\frac{p(1-p)}{\epsilon^2}\Bigr).
\end{align}
Now we give a sample complexity lower bound of algorithms solving \tsk{mean-estimation} with $m=1$ using \thm{low-depth}. The difficulty of directly applying \thm{low-depth} is that it requires $\dim{\mathrm{Im}(\Pi_c)}$ to be small. To resolve it, we prove that any quantum mean estimator can be modified to recover the state in query register $Q$ to $\ket{\mathbf{0}}$ with a small overhead so that correct answers lie in a much smaller subspace.

\begin{theorem}
    \label{thm:mean-estimation-lower-0}
    Suppose all random variables in \tsk{mean-estimation} have variance bounded by $\sigma^2$, and $|\mu|\le R$. Let $\mathcal{A}$ be a quantum query algorithm acting on registers $Q$, $W$ solving the quantum non-identical mean estimation problem defined in \tsk{mean-estimation} with repetition parameter $m=1$ and accuracy $\epsilon/2$. Suppose that $\frac{1}{2}n_Q>n_W+2\log(\frac{2R}{\epsilon})+1$, then it requires $T = \Omega(\frac{\sigma^2}{\epsilon^2})$ for the existence of such an algorithm $\mathcal{A}$, and $\mathcal{A}$ needs $T = \Omega\bigl(\frac{\sigma^2}{\epsilon^2}\bigr)$ quantum experiments.
\end{theorem}
\begin{proof}
Use the uncomputation trick to combine \thm{low-depth} and \thm{k-parallel-lower}. See \append{lower-0}.
\end{proof}

\subsubsection{Implication for Quantum Linear Systems}
\label{sec:implication_quantum_ls}

As mentioned in the introduction, we can possibly estimate $A$ by the following procedure.

For fixed integers $t_0, \gamma = \Theta(\log (\sqrt{n}/\delta))$ and any $0 \leq t < n$, suppose we have a register storing $|\psi_{t_0 + 2 \gamma t} \rangle$. We measure $|\psi_{t_0 + 2 \gamma t} \rangle$ to obtain a classical state $x_{t_0 + 2 \gamma t}$, and get $|\psi_{t_0 + 2 \gamma t + 1}\rangle$ as the second register of $U_f |\psi_{t_0 + 2 \gamma t}\rangle |0\rangle $ (note that $|\psi_{t_0 + 2 \gamma t} \rangle$ has collapsed after the measurement), which encodes the randomness of $x_{t_0 + 2 \gamma t + 1}$ given $x_{t_0 + 2 \gamma t}$. Similarly, we can also obtain $|\psi^{-1}_{t_0 + 2 \gamma t + 1}\rangle$ by querying $U^{-1}_f$. After that, we compute $U_f |\psi_{t_0 + 2 \gamma t + 1}\rangle |0\rangle $ and collect the second register as $|\psi_{t_0 + 2 \gamma t + 2}\rangle$, and do this computation for all $t_0 + 2 \gamma t + 1$ to $t_0 + 2 (\gamma + 1) t - 1 $. Then we let $t = t + 1$ and repeat this process.

After such process, we have $n$ classical samples at even steps $X_{t_0} := [x_{t_0}, x_{t_0 + 2\gamma}, \ldots, x_{t_0 + 2n \gamma - 2}] \in \mathbb{R}^{n \times n}$, and $n$ quantum samples at odd steps. It holds that $X_{t_0}$ is full rank with probability 1 given that
\begin{align}
\label{eqn:ax0}
A X_{t_0} = [x_{t_0 + 1}, \ldots, x_{t_0 + 2n \gamma - 1}] + W_{t_0} + Z_{t_0}
\end{align}
where $W_{t_0}$ is a zero-mean noise matrix and $\|Z_{t_0}\|_F \leq O(\delta)$. The matrix $Z_{t_0}$ denotes the difference between $\mathbb{E}[x_{t_0 + 2\gamma t + 1} \mid x_{t_0 + 2\gamma t}]$ and $\mathbb{E}[x_{t_0 + 2\gamma t + 1} \mid x_{t_0 + 2\gamma t}, x_{t_0 + 2 \gamma (t + 1)}]$, which are close since $\|A^n\|_2 = O(-\exp(n))$. We define the quantum unitary $U_{t_0}$ as 
\begin{align}
U_{t_0} |0 \rangle := \int_{W} \sqrt{f_{t_0}(W)} |\psi_{t_0 + 1}, \ldots, \psi_{t_0 + 2n \gamma - 1}\rangle  X_{t_0}^{-1}  \mathrm{d}W
\end{align}
where $f_{t_0}(W)$ is the pdf of random matrix $W_{t_0} X_{t_0}^{-1}$. Then we can use the quantum samples collected at steps $t_0 + 1, \ldots, t_0 + 2n \gamma - 1$ as the return of query to $U_{t_0}$ (or $U_{t_0}^{-1}$). Note that the mean of the random variable encoded by $U_{t_0}$ is $O(\delta)$-close to $A$ in Frobenius norm according to (\ref{eqn:ax0}). However, the distribution encoded in $U_{t_0}$ are different for different $t_0$ since $X_{t_0}$ are different. The lower bound presented in the previous section shows that this methods cannot achieve a desired quantum speed-up since the oracle $U_{t_0}$ can only be queried once for each $t_0$.

\subsection{Lower Bounds for $m\ge 1$}
\label{sec:lower-bound-po}
Given a boolean string $|x|\in \{0,1\}^n$ and $k\in [n]$, the task of distinguishing $|x|=k$ and $|x|=k+1$ or $|x|=k-1$ can be reduced to estimating $\frac{|x|}{n}$ to within $\frac{1}{n}$ additive error, which can be regarded as a mean estimation problem. Therefore, the query complexity lower bound for the first problem is also a lower bound for the second problem. As a result, we first prove the query complexity lower bound of the first problem given non-identical oracles.

We use the same quantum query algorithm model in \sec{lower-bound-sg}, where the algorithm pre-determines $U_0,\ldots,U_T$ and needs to distinguish the cases between $|x|=k$ and $|x|=k+1$ or $k-1$ for any $1\le k < n$.
\begin{lemma}
    \label{lem:counting-non-identical}
    Given a sequence of oracles $O_{x_1},\ldots,O_{x_T}$ encoding boolean strings $x_1,\ldots,x_T$ in $\{0,1\}^n$, suppose all strings have the same Hamming weight $w$ and the algorithm can query each oracle at most $m$ times in turn. For any $1\le k<n$ and $m=O(\sqrt{n})$, any quantum algorithm needs $\Omega(\frac{n}{m})$ queries in total to distinguish between $w=k$ and $w=k-1$ or $k+1$ with high probability.
\end{lemma}
\begin{proof}
    See \append{counting-non-identical}.
\end{proof}

Now we give a sample complexity lower bound of the quantum non-identical mean estimation problem with repetition parameter $m$.
\begin{theorem}
    \label{thm:mean-estimation-lower-1}
   Suppose all random variables in \tsk{mean-estimation} are Bernoulli random variables with mean $\mu \in(0,1)$ such that $\epsilon \le \mu(1-\mu)$ and $\epsilon = O(\frac{1}{m^2})$. It requires $T = \Omega(\frac{1}{\epsilon m^2})$ if there exists a quantum algorithm which queries each random variable at most $m$ times in turn solves this problem. Any such quantum query algorithm needs $mT = \Omega(\frac{1}{\epsilon m})$ quantum experiments in total.
\end{theorem}
\begin{proof}
    Let $n=\frac{1}{\epsilon}$ and $k=\mu n$. Since $\epsilon\le \mu(1-\mu)$, we have $1 \le  k \le  n-1$. Given a boolean string $|x|\in \{0,1\}^n$, the task of distinguishing $|x|=k$ and $|x|=k+1$ or $|x|=k-1$ can be reduced to estimating $\frac{|x|}{n}$ to within $\frac{1}{n}$ additive error. The latter problem can be regarded as estimating the mean of a Bernoulli random variable $X$ to within additive error $\epsilon=\frac{1}{n}$. Since one query to $O_X$ can be implemented by one query to $O_x$, the query complexity lower bound for the first problem is also a lower bound for the second problem. From $\epsilon=O(\frac{1}{m^2})$, we have $m=O(\frac{1}{\sqrt{\epsilon}})=O(\sqrt{n})$. Therefore, by \lem{counting-non-identical}, any quantum algorithm solving the quantum non-identical mean estimation problem with repetition parameter $m$ needs $\Omega(\frac{1}{\epsilon m})$ quantum experiments in total.
\end{proof}

\appendix

\section{Proof of the Upper Bound}

\subsection{Proof supplement of \thm{mean-bounded}}\label{append:mean-bounded}
In this section, we prove that \alg{bounded} outputs a mean estimation $\tilde{\mu}$ with additive error $O(\epsilon)$ with probability at least $2/3$.

Let $\epsilon'' = \frac{\epsilon}{H-L}$. By \lem{amplitude_estimation}, there exists a quantum circuit $C$ consisting of $M=O(\frac{1}{\epsilon''}) = O(\frac{H-L}{\epsilon})$ calls to controlled $S$ and $S^{\dagger}$ such that the measurement outcome of $C\ket{\mathbf{0}}$, denoted by $\tilde{p}$, satisfies
\begin{align}
\label{eq:proper-p}
\left|\tilde{p}- \frac{q^2}{2q^2-2q+1}\right| & \le \frac{2\pi q(1-q)}{M(2q^2-2q+1)}+\frac{\pi^2}{M^2}\\ 
& \le \frac{4\pi q(1-q)}{M}+\frac{\pi^2}{M^2} \\
& = O(q(1-q)\epsilon'' + \epsilon''^2) \\
& = O(q(1-q)\epsilon'')
\end{align}
for sufficiently small $\epsilon''=O(q(1-q))$, and measuring $C\ket{\mathbf{0}}$ gives such $y$ with probability at least $\frac{8}{\pi^2}$.

Replacing all the controlled $S$ and $S^{\dagger}$ with controlled $S_i$ and $S_i^{\dagger}$ gives a quantum circuit $C'$. Note for any two unitary $U,V$, we have \begin{align}
  \|\ket{0}\bra{0}\otimes I+ \ket{1}\bra{1}\otimes U-(\ket{0}\bra{0}\otimes I+ \ket{1}\bra{1}\otimes V)\| \le \|U-V\|,
\end{align} 
and \begin{align}
  \|\ket{\phi}\bra{\phi}-\ket{\phi_i}\bra{\phi_i}\| & \le \|\ket{\phi}\bra{\phi}-\ket{\phi_i}\bra{\phi}\|+ \|\ket{\phi_i}\bra{\phi}- \ket{\phi}\bra{\phi}\|\\
  & \le 2\|\ket{\phi}-\ket{\phi_i}\|\\
  & \le 2\sqrt{\Bigl(2\sqrt{2}\frac{\delta}{H-L}\Bigr)^2+\Bigl(2\sqrt{2}\frac{\delta}{H-L}\Bigr)^2} = O\left(\frac{\delta}{H-L}\right), \label{eqn:delta_1}
\end{align}
where the last inequality holds since $q_i$ is $\frac{\delta}{H-L}$-close to $q$ and we have
\begin{align}
\Big|\dv{x}(\frac{x}{\sqrt{2x^2-2x+1}})\Big|\le 2\sqrt{2},\qquad \Big|\dv{x}(\frac{1-x}{\sqrt{2x^2-2x+1}})\Big|\le 2\sqrt{2}
\end{align}
for all $x\in [0,1]$.
Therefore, by \lem{approx}, it holds that
\begin{align}
 \|C-C'\|
& \leq M \max_{i \in [T]} \|I-2S\ket{\mathbf{0}}\bra{\mathbf{0}}S^{\dagger}-(I-2S_i\ket{\mathbf{0}}\bra{\mathbf{0}}S_i^{\dagger})\| \\
& = 2M\max_{i\in [T]} \|S\ket{\mathbf{0}}\bra{\mathbf{0}}S^{\dagger}-S_i\ket{\mathbf{0}}\bra{\mathbf{0}}S_i^{\dagger}\|\\
& \leq 2M\left(2\epsilon_i+2\sqrt{\epsilon_i(1-\epsilon_i)}+\max_{i\in[T]}\|\ket{\phi}\bra{\phi}-\ket{\phi_i}\bra{\phi_i}\|\right) \\
& = O\left(M(\sqrt{\epsilon'}+\delta)\right)=  O\left(\frac{H-L}{\epsilon}\frac{\epsilon}{H-L}\right)\label{eqn:delta_2}\\
& = O(1).
\end{align}
where the third line uses Eq.~\eq{Si} and Eq.~\eq{S}.
Hence, $||C\ket{\mathbf{0}} - C'\ket{\mathbf{0}}|| = O(1)$ and the measurement of $C'\ket{\mathbf{0}}$ gives $\tilde{p}$ satisfying Eq.~\eq{proper-p} with probability at least $\frac{8}{\pi^2} - O(1)$. By adjusting the constant in $\epsilon'=O(\frac{\epsilon^2}{(H-L)^2})$, we can make the success probability be at least $\frac{2}{3}$.

Let $\tilde{q} = \frac{\tilde{p}-\sqrt{\tilde{p}(1-\tilde{p})}}{2\tilde{p}-1}$ be the estimation of $q$ and $p=\frac{q^2}{2q^2-2q+1}$. By Taylor's theorem 
\begin{align}
  \tilde{q}&= q+\frac{(q^2+(q-1)^2)^2}{2q(1-q)}(\tilde{p}-p)+O((\tilde{p}-p)^2)\\
  &= q+O\Bigl(\frac{q(1-q)\epsilon''}{q(1-q)}\Bigr)+O((q(1-q)\epsilon'')^2) \\
  &= q+O(\epsilon''),
\end{align}
where the second equality is obtained by Eq.~\eq{proper-p}.
Let $\tilde{q}(H-L)+L$ be the final output of the algorithm, then we have \begin{align}
  |\tilde{\mu} - \mu| = |\tilde{q}(H-L)+L - \mu| = |(\tilde{q}-q)(H-L)| = O(\epsilon)
\end{align}
with probability at least $\frac{2}{3}$.

\subsection{Proof of \lem{mean-sub-gaussian}}\label{append:mean-sub-gaussian}
In this section, We give a detailed proof for \lem{mean-sub-gaussian}.
\begin{lemma}[\lem{mean-sub-gaussian}]
Suppose all random variables $X_1,\ldots, X_T$ in \tsk{mean-estimation} are sub-Gaussian with parameter $K$ and their mean satisfies $|\mu_i| \le R$, $R \leq K$. Let $m,R,K,\epsilon$, $\delta$ in \alg{bounded_Gaussian} satisfies that $m =\Omega\Bigl(\log\Big(\frac{K\sqrt{\log(\frac{K}{\epsilon})}}{\epsilon}\Big)\Bigr)$, $\epsilon=O(K)$, and $\delta < \epsilon/4$. \alg{bounded_Gaussian} solves  \tsk{mean-estimation} if $T = \Omega(\frac{K\sqrt{\log (\frac{K}{\epsilon})}}{\epsilon})$, using $O\Bigl(\frac{K\sqrt{\log(\frac{K}{\epsilon})}}{\epsilon}\log\Big(\frac{K\sqrt{\log(\frac{K}{\epsilon})}}{\epsilon}\Big)\Bigr)$ quantum experiments in total.
\end{lemma}
\begin{proof}
For any $i \in [T]$, $O_{\tilde{X}_i}$ generates a quantum random variable truncated by $X_i$. Let $\tilde{X}_i$ be the truncated version of $X_i$ such that \begin{align}
    \tilde{X}_i = \begin{cases} X_i & X_i\in [L,H] \\ 0 & \text{otherwise}.\end{cases}
  \end{align}
$O_{\tilde{X}_i}$ can be seen as a quantum random variable generating $\tilde{X}_i$.

Now we give a bound on the difference between the mean of $X_i$ and $\tilde{X}_i$.  We first present a well-known tail bound of Gaussian random variables. Since $\epsilon = O(K)$, it holds that $\Delta \ge K\sqrt{4\log(\frac{128K}{\epsilon})}\ge K$. For any $x>0$, we have
  \begin{align}
    \label{eq:tail-bound}
    \int_{x}^{+\infty} \exp\Bigl(-\frac{t^2}{2K^2}\Bigr)\, \d t \le \int_{x}^{+\infty} \frac{t}{x}\exp\Bigl(-\frac{t^2}{2K^2}\Bigr)\, \d t 
    =  \frac{K^2}{x}\exp\Bigl(-\frac{x^2}{2K^2}\Bigr).
  \end{align}

  For all $i\in [T]$, we have
  \begin{align}
    &|\mathbb{E}[X_i]-\mathbb{E}[\tilde{X}_i]|\nonumber\\ 
    \le & \Big|\int_{-\infty}^{L}tp_i(t)\, \d t| + |\int_{H}^{\infty}tp_i(t)\, \d t\Big|\\
    = & \Bigl|L\mathbb{P}[X_i\le L] - \int_{-\infty}^{L}\mathbb{P}[X_i\le t]\, \d t\Bigr|+ \Bigl|H\mathbb{P}[X_i\ge H] + \int_{H}^{\infty}\mathbb{P}[X_i\ge t]\, \d t\Bigr| && \text{(by integration by parts)}\\
    \le & 2(L+\frac{K^2}{\mu_i-L})\exp\Bigl(-\frac{(L-\mu_i)^2}{2K^2}\Bigr) + 2(H+\frac{K^2}{H-\mu_i})\exp\Bigl(-\frac{(H-\mu_i)^2}{2K^2}\Bigr) && \text{(by Eq.~\eq{def-sub-gaussian} and Eq.~\eq{tail-bound})} \\
    \le & 4(\Delta + R + \frac{K^2}{\Delta})\exp\Bigl(-\frac{\Delta^2}{2K^2}\Bigr) && \text{(by }|\mu_i|\le R)\\
    \le & 4(2\Delta +R)\exp\Bigl(-\frac{\Delta^2}{2K^2}\Bigr) && \text{(by }\Delta \ge K)\\
    \le & 4\Bigl(2K\sqrt{4\log\Bigl(\frac{128K}{\epsilon}\Bigr)}\Bigl(\frac{\epsilon}{128K}\Bigr)^2+R\frac{\epsilon}{32R}\Bigr) \label{eq:xexp-decrease}
    =  \frac{\epsilon}{4}\sqrt{\log\Bigl(\frac{128K}{\epsilon}\Bigr)}\frac{\epsilon}{128K}+\frac{\epsilon}{8} \\
    \le & \frac{\epsilon}{4}, && \text{(by } \sqrt{\log(x)}\le x)
  \end{align}
  where Eq.~\eq{xexp-decrease} holds since $x\exp(-\frac{x^2}{2K^2})$ decreases for $x \ge K$. Then we have \begin{align}
    |\mathbb{E}[\tilde{X}_i]-\mu| \le |\mathbb{E}[\tilde{X}_i]-\mu_i|+|\mu_i-\mu| \le \delta + \frac{\epsilon}{4}\le \frac{\epsilon}{2}.
  \end{align}

  Since $\tilde{X}_i$ are all bounded random variable in $[L,H]$ with $|\mathbb{E}[\tilde{X}_i] - \mu|\leq \epsilon$ and $m=\Omega\Bigl(\log\Big(\frac{K\sqrt{\log(\frac{K}{\epsilon})}}{\epsilon}\Big)\Bigr)$, $ \epsilon=O(K)=O(\Delta)=O\big(\frac{(R+\Delta-\mu)(R+\Delta+\mu)}{R+\Delta}\big)$,  by \thm{mean-bounded} we can conclude that $\tilde{\mu}$ is an estimation to $\mu$ to within additive error $O(\epsilon)$ with probability at least $\frac{2}{3}$ using $O\Bigl(\frac{K\sqrt{\log(\frac{K}{\epsilon})}}{\epsilon}\log\Big(\frac{K\sqrt{\log(\frac{K}{\epsilon})}}{\epsilon}\Big)\Bigr)$ quantum experiments in total. 
\end{proof}

\section{Proof of the Lower Bound}
\subsection{Proof of \sta{tement}}\label{append:statement}
In this section, we give a detailed proof of \sta{tement}.
\begin{statement}[\sta{tement}]
        Let $O_X$ be any unitary satisfying Eq.~\eqn{def-query}, and $U_0^{\mathrm{low}},\ldots,U_T^{\mathrm{low}}$ be a sequence of quantum circuits satisfying $U_{0}^{\mathrm{low}}=I$ and \begin{align}\label{eqn:recursion-U-low-append}
                U_{t+1}^{\mathrm{low}} = \begin{cases}
                    ((U_t)_{Q_t,W} \otimes I) \cdot (U_{t}^{\mathrm{low}}\otimes (O_X)_{Q_{t+1}}) & \text{if } a_{t+1}=1, \\
                    ((U_t)_{Q_t,W} \otimes I) \cdot (U_{t}^{\mathrm{low}}\otimes I_{Q_{t+1}}) & \text{if } a_{t+1}=-1,
                \end{cases}
            \end{align}
            for all $0\le t < T$. The quantum circuit $U_{t}^{\mathrm{low}}$ can prepare $|\psi_{\mathrm{eff}}^{(t)}\rangle$ after post-selection, namely, \begin{align}
        \label{eqn:def-U-low-append}
    |psi_{\mathrm{eff}}^{(t)}\rangle = \Bigl(I_{W,Q_t}\bigotimes_{i=1}^{t} \langle\phi_{\mathrm{beg}}^{(i)}|_{Q_{i-1}}\Bigr)U_t^{\mathrm{low}}\ket{\mathbf{0}}_{W,Q_0,\ldots,Q_t},
    \end{align} 
    for any $0\le t\le T$.
    \end{statement}
\begin{proof}
We prove this statement by induction on $t$. For $t=0$, we have \begin{align}
        U_{0}^{\mathrm{low}}\ket{\mathbf{0}}_W\ket{\mathbf{0}}_{Q_0} = \ket{\mathbf{0}}_W\ket{\mathbf{0}}_{Q_0} = |\psi_{\mathrm{eff}}^{(0)}\rangle_{Q_0,W},
    \end{align}
    which satisfies Eq.~\eqn{def-U-low-append}. Assume the statement is true for some $t \ge 0$. If $a_{t+1}=1$, we have \begin{align}
         &\Bigl(I_{W,Q_{t+1}}\bigotimes_{i=1}^{t+1} \langle\phi_{\mathrm{beg}}^{(i)}|_{Q_{i-1}}\Bigr)U_{t+1}^{\mathrm{low}}\ket{\mathbf{0}}_W\ket{\mathbf{0}}_{Q_0}\cdots \ket{\mathbf{0}}_{Q_{t+1}} \\
         =& \Bigl(\langle \phi_{\mathrm{beg}}^{(t+1)}|_{Q_t}(U_t)_{Q_t,W}\bigl(I_{W,Q_t}\bigotimes_{i=1}^{t} \langle\phi_{\mathrm{beg}}^{(i)}|_{Q_{i-1}}\bigr)U_{t}^{\mathrm{low}}\ket{\mathbf{0}}_{Q_0}\ket{\mathbf{0}}_W\cdots\ket{\mathbf{0}}_{Q_{t}}\Bigr)|\psi_X\rangle_{Q_{t+1}} && \text{(by Eq.~\eqn{recursion-U-low})} \\
         =&\Bigl(\langle \phi_{\mathrm{beg}}^{(t+1)}|_{Q_t}U_t|\psi^{(t)}_{\mathrm{eff}}\rangle_{Q_t,W}\Bigr)|\phi_{\mathrm{end}}^{(t+1)}\rangle_{Q_{t+1}} && \hspace{-22.5mm} \text{(by Eq.~\eqn{def-U-low-append} and Eq.~\eqn{def-end})} \\
         =&\bigl(\langle \phi_{\mathrm{beg}}^{(t+1)}|_{Q_t}U_t|\phi_{\mathrm{end}}^{(t)}\rangle_{Q_t}|\phi_{\mathrm{W}}^{(t)}\rangle_{W}\bigr)|\phi_{\mathrm{end}}^{(t+1)}\rangle_{Q_{t+1}} && \text{(by Eq.~\eqn{eff-pro})}\\
         =& |\phi_{W}^{(t+1)}\rangle_{W}|\phi_{\mathrm{end}}^{(t+1)}\rangle_{Q_{t+1}} && \text{(by Eq.~\eqn{eff-recursion})}\\
         =&|\psi_{\mathrm{eff}}^{(t+1)}\rangle_{Q_{t+1},W}. && \text{(by Eq.~\eqn{eff-pro})}
    \end{align}
    The proof for $a_{t+1}=-1$ is basically the same except for the state in register $Q_{t+1}$. Therefore, $U_{t}^{\mathrm{low}}$ constructed in Eq.~\eqn{recursion-U-low} satisfies Eq.~\eqn{def-U-low-append} for all $0\le t\le T$.
    
    Expand Eq.~\eqn{recursion-U-low-append}, we have \begin{align}
        U_{t}^{\mathrm{low}} = \prod_{i=0}^{t-1} (U_i)_{Q_{i},W} \bigotimes_{1\le i\le t, a_i=1} (O_X)_{Q_i},
    \end{align}
    which has query depth 1. In conclusion, the statement is true for all $0\le t \le T$. 
\end{proof}

\subsection{Proof of \lem{can-control}}\label{append:can-control}
In this section, we give a detailed proof of \lem{can-control}.
\begin{lemma}[\lem{can-control}]
     For any $T$-query quantum algorithm acting on registers $Q$, $W$, and any finite random variable $X$ on $(\Omega,p)$, if $\dim{\mathcal{H}_Q}>2\dim{\mathcal{H}_W}$, there exists a sequence of quantum random variables  $(\mathcal{H}_Q,O_X^{(1)}), \ldots, (\mathcal{H}_Q,O_X^{(T-1)})$ such that for any $0 \le t < T$ \begin{align}
         |\psi^{(t)}\rangle = |\phi_{\mathrm{beg}}^{(t+1)}\rangle|\phi_W^{(t+1)}\rangle + |\psi^{(t)}_{\perp}\rangle,
     \end{align}
     for some unnormalized state $|\psi^{(t)}_{\perp}\rangle$ orthogonal to $|\phi_{\mathrm{beg}}^{(t+1)}\rangle\otimes\mathcal{H}_W$.
\end{lemma}
\begin{proof}

We prove this lemma by induction on $t$.
    
    We first prove the case for $t=0$. From Eq.~\eqn{eff-recursion}, we have \begin{align}
        |\phi_{\mathrm{W}}^{(1)}\rangle = (\langle\phi_{\mathrm{beg}}^{(1)}|\otimes I)U_{0}\ket{\mathbf{0}}\ket{\mathbf{0}},
    \end{align}
    which means \begin{align}
        (\bra{\phi_{\mathrm{beg}}^{(1)}}\otimes I)|\psi^{(0)}\rangle = (\bra{\phi_{\mathrm{beg}}^{(1)}}\otimes I)U_0\ket{\mathbf{0}}\ket{\mathbf{0}}=|\phi_{\mathrm{W}}^{(1)}\rangle,
    \end{align}
    so Eq.~\eqn{recur-control} holds for $t=0$.
    
    Assuming Eq.~\eqn{recur-control} is true for some $t\ge0$, we then prove the case for $t+1$. Note that \begin{align}
        &(\langle \phi_{\mathrm{beg}}^{(t+2)}|\otimes I)|\psi^{(t+1)}\rangle \\
        &\qquad = (\langle \phi_{\mathrm{beg}}^{(t+2)}|\otimes I)U_{t+1}(O_X^{(t+1)})^{a_{t+1}}|\psi^{(t)}\rangle \\
        &\qquad =(\langle \phi_{\mathrm{beg}}^{(t+2)}|\otimes I)\bigl(U_{t+1}|\phi_{\mathrm{end}}^{(t+1)}\rangle|\phi_W^{(t+1)}\rangle + U_{t+1}(O_X^{(t+1)})^{a_{t+1}}|\psi^{(t)}_{\perp}\rangle\bigr) && \text{(by Eq.~\eqn{recur-control})}\\
        &\qquad =|\psi^{(t+1)}_W\rangle + (\langle \phi_{\mathrm{beg}}^{(t+2)}|\otimes I)U_{t+1}(O_X^{(t+1)})^{a_{t+1}}|\psi^{(t)}_{\perp}\rangle. && \text{(by Eq.~\eqn{eff-recursion})}
    \end{align}
    To make Eq.~\eqn{recur-control} hold for $t+1$, we need to find a unitary operator $O_X^{(t+1)}$ such that \begin{align}
    \label{eqn:success-flag}
    (O_X^{(t+1)})^{a_{t+1}}|\psi^{(t)}_{\perp}\rangle\in (U_{t+1}^{\dagger}(|\phi^{(t+2)}_{\mathrm{beg}}\rangle\otimes \mathcal{H}_W))^{\perp}.
    \end{align}
    Since $\dim{\mathcal{H}_Q}>2\dim{\mathcal{H}_W}>\dim{\mathcal{H}_W}$, the Schmidt decomposition of $|\psi_{\perp}^{(t)}\rangle$ is \begin{align}
|\psi_{\perp}^{(t)}\rangle=\sum_{i=1}^{\dim{\mathcal{H}_W}}\lambda_i |i_Q\rangle_Q|i_W\rangle_W,
    \end{align}
    where $\{|i_Q\rangle\}$ and $\{|i_W\rangle\}$ are two orthonormal set of states. By induction hypothesis, $|\psi^{(t)}_{\perp}\rangle$ is orthogonal to $|\phi_{\mathrm{beg}}^{(t+1)}\rangle\otimes\mathcal{H}_W$, so all $|i_Q\rangle$ are orthogonal to $|\phi_{\mathrm{beg}}^{(t+1)}\rangle$.

    Note that $(O_X^{(t+1)})^{a_{t+1}}|\phi_{\mathrm{beg}}^{(t+1)}\rangle=|\phi_{\mathrm{end}}^{(t+1)}\rangle$ and $(O_X^{(t+1)})^{a_{t+1}}$ is unitary, hence by controlling $O_X^{(t+1)}$, $(O_X^{(t+1)})^{a_{t+1}}|\psi_{\perp}^{(t)}\rangle$ can be \begin{align}
        \label{eqn:schmidt-decompose-after}(O_X^{(t+1)})^{a_{t+1}}|\psi_{\perp}^{(t)}\rangle=\sum_{i=1}^{\dim{\mathcal{H}_W}}\lambda_i |i_{Q}'\rangle_Q|i_W\rangle_W
    \end{align}
    for any orthonormal set of states $\{|i_Q'\rangle\}$ in $|\phi_{\mathrm{end}}^{(t+1)}\rangle^{\perp}$.

    To make Eq.~\eqn{success-flag} hold, we try to construct $|i_Q'\rangle$ successively so that they are in $(U_{t+1}^{\dagger}(|\phi^{(t+1)}_{\mathrm{beg}}\rangle\otimes \mathcal{H}_W))^{\perp}$ and form an orthonormal set of states. We give the construction by induction. Assume we have constructed the first $k-1$ states $(|i'\rangle_Q)_{i=1}^{k-1}$, the possible subspace of $|k'_Q\rangle_Q|k_W\rangle_W$ is \begin{align}
        (\mathrm{span}\{(|i'_Q\rangle)_{i=1}^{k-1},|\phi_{\mathrm{end}}^{(t+1)}\rangle\})^{\perp}\otimes |k_W\rangle_W,
    \end{align}
    which has dimension $\dim{\mathcal{H}_Q}-k$.  Since 
    \begin{align}
        &\dim\bigl(\mathrm{span}\{(|i'_Q\rangle)_{i=1}^{k-1},|\phi_{\mathrm{end}}^{(t+1)}\rangle\})^{\perp}\otimes |k_W\rangle_W\bigr) + \dim \bigl((U_{t+1}^{\dagger}(|\phi^{(t+1)}_{\mathrm{beg}}\rangle\otimes \mathcal{H}_W))^{\perp}\bigr) \\
        &\qquad= \dim{\mathcal{H}_Q}-k + \dim(\mathcal{H}_Q\otimes \mathcal{H}_W) - \dim{\mathcal{H}_W} \\ 
        &\qquad\ge\dim{\mathcal{H}_Q}-\dim{\mathcal{H}_W} + \dim(\mathcal{H}_Q\otimes \mathcal{H}_W) - \dim{\mathcal{H}_W}\\
        &\qquad\geq \dim(\mathcal{H}_Q\otimes \mathcal{H}_W),
    \end{align}
    where the last inequality comes from the assumption $\dim{\mathcal{H}_Q}>2\dim{\mathcal{H}_W}$, we can deduce that the intersection of these two subspaces is non-empty. Hence we can find an normalized state  $|k_Q'\rangle_Q|k_W\rangle_W$ in this intersection space. By induction, we can construct orthonormal states $(|i'_Q\rangle)_{i=1}^{\dim{\mathcal{H}_W}}$ in $|\phi_{\mathrm{end}}^{(t+1)}\rangle^{\perp}$ so that $|i_{Q}'\rangle_Q|i_W\rangle_W\in (U_{t+1}^{\dagger}(|\phi^{(t+1)}_{\mathrm{beg}}\rangle\otimes \mathcal{H}_W))^{\perp}$ for all $i\in[\dim{\mathcal{H}_W}]$ . From Eq.~\eqn{schmidt-decompose-after}, there exists a unitary operator $O_X^{(t+1)}$ such that Eq.~\eqn{success-flag} is true.   
\end{proof}

\subsection{Proof of \thm{mean-estimation-lower-0}}\label{append:lower-0}
In this section, we give a detailed proof of \thm{mean-estimation-lower-0}.
\begin{theorem}[\thm{mean-estimation-lower-0}]
    Suppose all random variables in \tsk{mean-estimation} have variance bounded by $\sigma^2$, and $|\mu|\le R$. Let $\mathcal{A}$ be a quantum query algorithm acting on registers $Q$, $W$ solving the quantum non-identical mean estimation problem defined in \tsk{mean-estimation} with repetition parameter $m=1$ and accuracy $\epsilon/2$. Suppose that $\frac{1}{2}n_Q>n_W+2\log(\frac{2R}{\epsilon})+1$, then it requires $T = \Omega(\frac{\sigma^2}{\epsilon^2})$ for the existence of such an algorithm $\mathcal{A}$, and $\mathcal{A}$ needs $T = \Omega\bigl(\frac{\sigma^2}{\epsilon^2}\bigr)$ quantum experiments.
\end{theorem}
\begin{proof}
 Let $\mathbb{R}_{\epsilon}=\{i\epsilon\mid i\in \mathbb{Z}\}$ be an $\epsilon$-net of $\mathbb{R}$. Denote the output of $\mathcal{A}$ be $\tilde{\mu}$, and let $\mu_{\epsilon}$ be the closest number to $\tilde{\mu}$ in $\mathbb{R}_{\epsilon}$. We can delay the measurement of $\mathcal{A}$ and compute $\mu_{\epsilon}$ in an additional register $W_1$ with $n_{W_1}=\log(\frac{2R}{\epsilon})$ coherently which gives a unitary $U$ such that \begin{align}
        U\ket{\mathbf{0}}_{Q,W,W_1}=\sum_{i\in {\mathbb{Z}}}\sqrt{p(i)}\ket{\phi_i}_{Q,W}\ket{i\epsilon}_{W_1}
    \end{align}
    for some distribution $p$ and unit states $\ket{\phi_i}$. Let the two closest number in $\mathbb{R}_{\epsilon}$ to the true mean $\mu$ be $i^{*}\epsilon$ and $(i^{*}+1)\epsilon$. Since $\tilde{\mu}$ is an $\epsilon/2$-additive approximation of $\mu$ with probability $2/3$, $\mu_{\epsilon}$ equals $i^{*}\epsilon$ or $(i^{*}+1)\epsilon$ with probability at least $2/3$. Therefore, $p(i^*)+p(i^*+1) \ge 2/3$ and hence \begin{align}
        p(i^*)^2+p(i^*+1)^2\ge\frac{2}{9}.
    \end{align}
    Appending another register $W_2$ with $n_{W_2}=n_{W_1}$ and using the same technique in \thm{mean-bounded}, we can uncompute the state in $Q,W,W_1$ by the following unitary. Let $V=(U^{\dagger}\otimes I)(I_{Q,W}\otimes \mathrm{CNOT}_{W_1,W_2})(U\otimes I)$, then we have \begin{align}
        V\ket{\mathbf{0}}_{Q,W,W_1,W_2}=\ket{\mathbf{0}}_{Q,W,W_1}\sum_{i\in \mathbb{Z}}p(i)\ket{i\epsilon}_{W_2}+\ket{\mathrm{garbage}}
    \end{align}
    for some unknown garbage state $\ket{\mathrm{garbage}}$ orthogonal to $\ket{\mathbf{0}}_{Q,W,W_1}$. 
    
    By \lem{fixsearch}, we can prepare $\ket{\mathbf{0}}_{Q,W,W_1}\sum_{i\in \mathbb{Z}}p(i)\ket{i\epsilon}_{W_2}$ with high probability using 
    \begin{align}
    O\left(\frac{1}{\sqrt{\sum_{i\in \mathbb{Z}}p(i)^2}}\right)=O\left(\frac{1}{\sqrt{p(i^*)^2+p(i^*+1)^2}}\right)=O(1)
    \end{align}
    calls to $V$, and then measuring the state in register $W_2$ gives an $\epsilon$-additive approximation of $\mu$ with probability at least $2/9$. Denote this algorithm by $\mathcal{A}'$. $\mathcal{A}'$ has query register $Q'$ with $n_{Q'}=n_Q$ and working register $W'$ with $n_{W'}=n_W+n_{W_1}+n_{W_2}=n_W+2\log(\frac{2R}{\epsilon})$. Assume that $\mathcal{A}$ uses $T$ queries, and then $\mathcal{A}'$ uses $T'=O(T)$ queries since $\mathcal{A}'$ calls $\mathcal{A}$ $O(1)$ times. 

    Let $\Pi_c=\ket{\mathbf{0}}\bra{\mathbf{0}}_{Q,W,W_1}\otimes(\ket{i^*\epsilon}\bra{i^*\epsilon}_{W_2}+\ket{(i^*+1)\epsilon}\bra{(i^*+1)\epsilon}_{W_2})$, which is a projection onto a 2-dimensional space with correct output, and let $\ket{\psi^{(T')}}$ be the final state of $\mathcal{A}'$ before measurements. Since $\dim \mathrm{Im}(\Pi_c) < \dim \mathcal{H}_{W'} <\frac{1}{2} \dim \mathcal{H}_{Q'}$, by \thm{low-depth}, there exists another quantum circuit $U^{\mathrm{low}}$ using two $T'$-parallel queries and a sequence of quantum random variables $(\mathcal{H}_Q,O_X^{(1)}),$ $\ldots,$ $ (\mathcal{H}_Q,O_X^{(T')})$ satisfying
    \begin{align}
        \|\bigl(\Pi_c\otimes \langle\mathbf{0}|_{Q_0,\ldots,Q_{T'-1}}\bigr)U^{\mathrm{low}}\ket{\mathbf{0}}_{W,Q_0,\ldots,Q_{T'}}\|^2=\|\Pi_c\ket{\psi^{(T')}}\|^2\approx \sqrt{\frac{p(i^*)^2+p(i^*+1)^2}{\sum_{i\in \mathbb{Z}}p(i)^2}}=\Omega(1)
        \end{align}  
    where $Q_0,\ldots,Q_{T'}$ are $T'+1$ registers with $n_{Q'}$ qubits. Therefore, by applying $U^{\mathrm{low}}$ and measuring the final state, we can estimate the mean of $X$ using two $T'$-parallel queries to any $O_X$ encoding $X$ with constant success probability. The construction of $U^{\mathrm{low}}$ in \thm{low-depth} is independent of $X$, so it can be applied to estimate the mean of any $X$ with bounded variance. We consider the case that $X$ is a Bernoulli random variable. By Eq.~\eqn{flip-to-state}, $O_X$ can be simulated by one query to $O_x$. Therefore, by \thm{k-parallel-lower}, $T'$ needs to be $\Omega\bigl(\frac{\mu(1-\mu)}{\epsilon^2}\bigr)=\Omega\bigl(\frac{\sigma^2}{\epsilon^2}\bigr)$ so that $\mathcal{A}'$ can estimate $\mu$ to within $\epsilon$ additive error using two $T'$-parallel queries to $O_X$. Since $T'=O(T)$, we have $T=\Omega(T')=\Omega\bigl(\frac{\sigma^2}{\epsilon^2}\bigr)$.
    \end{proof}

\subsection{Proof of \lem{counting-non-identical}}\label{append:counting-non-identical}
In this section, we give a detailed proof of \lem{counting-non-identical}.
\begin{lemma}[\lem{counting-non-identical}]
    Given a sequence of oracles $O_{x_1},\ldots,O_{x_T}$ encoding boolean strings $x_1,\ldots,x_T$ in $\{0,1\}^n$, suppose all strings have the same Hamming weight $w$ and the algorithm can query each oracle at most $m$ times in turn. For any $1\le k<n$ and $m=O(\sqrt{n})$, any quantum algorithm needs $\Omega(\frac{n}{m})$ queries in total to distinguish between $w=k$ and $w=k-1$ or $k+1$ with high probability.
\end{lemma}
\begin{proof}
We construct two string sequences iteratively, which are hard to be distinguished by the algorithm. The Hamming weights of the strings are $k$ in one sequence and $k+1$ in the other sequence. 

Let $t$ be any multiple of $m$ so that the algorithm will query a new string at $t+1$. Let $|\psi_k^{(t)}\rangle$ be the state of the algorithm after querying the oracle $t$ times with all query strings have Hamming weight $k$. Similarly, let $|\psi_{k+1}^{(t)}\rangle$ be the state with all query strings having Hamming weight $k+1$.  These states are dependent on the strings that the algorithm queries prior to time $t$. However, since subsequent construction does not depend on the previous queries, we omit the subscripts indicating prior query strings for convenience.

Let $s\in\{0,1\}^n$ be any string with $|s|=k$ and $F_{s} = \{i\in [n]\mid s_i=0\}$. For any $i\in F_s$, let $s^{(i)}\in \{0,1\}^n$ be the same string as $s$ except for $s^{(i)}_i=1$ so $|s^{(i)}|=k+1$.

For any $l=0,\ldots,m$, let \begin{align}
    |\psi_k^{(t+l)}\rangle&=U_{t+l}O_s\cdots U_{t+1}O_s|\psi_k^{(t)}\rangle=\sum_{j\in[n]}\alpha_{j}^{(t+l)}|j\rangle|\phi_k^{(t+l)}\rangle,\\
    |\overline{\psi}_k^{(t+l)}\rangle&=U_{t+l}O_s\cdots U_{t+1}O_s|\psi_{k+1}^{(t)}\rangle=\sum_{j\in[n]}\overline{\alpha}_{j}^{(t+l)}|j\rangle|\overline{\phi}_k^{(t+l)}\rangle,
\end{align}
where the first register on the right side of the equation contains the first $\lceil \log_2 n\rceil$ qubits of the query register.

For any $i\in F_s$, let \begin{align}
|\psi_{k+1,i}^{(t+l)}\rangle & = U_{t+l}O_{s^{(i)}}\cdots U_{t+1}O_{s^{(i)}}|\psi_{k+1}^{(t)}\rangle=\sum_{j\in[n]}\beta_{i,j}^{(t+l)}|j\rangle|\phi_{k+1,i}^{(t+l)}\rangle.
\end{align}
We first prove that $|\overline{\psi}_k^{(t+l)}\rangle$ is an approximation of $|\psi_{k+1,i}^{(t+l)}\rangle$. Note that \begin{align}
    |\overline{\psi}_k^{(t)}\rangle = |\psi_{k+1}^{(t)}\rangle=|\psi_{k+1,i}^{(t)}\rangle
\end{align} and 
\begin{align}
    \||\overline{\psi}_k^{(t+l+1)}\rangle-|\psi_{k+1,i}^{(t+l+1)}\rangle\|_2 &=\|U_{t+l+1}O_s|\overline{\psi}_k^{(t+l)}\rangle-U_{t+l+1}O_{s^{(i)}}|\psi_{k+1,i}^{(t+l)}\rangle\|_2\\
    &=\|O_s|\overline{\psi}_k^{(t+l)}\rangle-O_{s^{(i)}}|\psi_{k+1,i}^{(t+l)}\rangle\|_2\\
    &=\|O_{s^{(i)}}O_s|\overline{\psi}_k^{(t+l)}\rangle-|\psi_{k+1,i}^{(t+l)}\rangle\|_2\\
    &\le \|O_{s^{(i)}}O_s|\overline{\psi}_k^{(t+l)}\rangle-|\overline{\psi}_k^{(t+l)}\rangle\|_2+\||\overline{\psi}_k^{(t+l)}\rangle-|\psi_{k+1,i}^{(t+l)}\rangle\|_2\\
    &=2|\overline{\alpha}_{i}^{(t+l)}|+\||\overline{\psi}_k^{(t+l)}\rangle-|\psi_{k+1,i}^{(t+l)}\rangle\|_2.
\end{align}
Therefore, by induction, we have \begin{align}
    \label{eqn:bar-approx}
    \||\overline{\psi}_k^{(t+l+1)}\rangle-|\psi_{k+1,i}^{(t+l+1)}\rangle\|_2\le 2\sum_{j=0}^{l}|\overline{\alpha}_{i}^{(t+j)}|.
\end{align}
Let \begin{align}
    S^{(l)} &= \sum_{i\in F_s}\langle \psi_{k}^{(t+l)}|\psi_{k+1,i}^{(t+l)}\rangle,
\end{align}
then the progress at $t+l+1$ satisfies \begin{align}
    S^{(l)}-S^{(l+1)}&=\sum_{i\in F_s}\langle \psi_{k}^{(t+l)}|\psi_{k+1,i}^{(t+l)}\rangle-\langle \psi_{k}^{(t+l+1)}|\psi_{k+1,i}^{(t+l+1)}\rangle\\
    &=\sum_{i\in F_s}\langle \psi_{k}^{(t+l)}|(I-O_sO_{s^{(i)}})|\psi_{k+1,i}^{(t+l)}\rangle\\
    &=2\sum_{i\in F_s}\langle\psi_{k}^{(t+l)}|\beta_{i,i}^{(t+l)}|i\rangle|\phi_{k+1,i}^{(t+l)}\rangle\\
    &=2\sum_{i\in F_s}(\alpha_{i}^{(t+l)})^{*}\beta_{i,i}^{(t+l)}\langle \phi_k^{(t+l)}|\phi_{k+1,i}^{(t+l)}\rangle,
\end{align}
which implies \begin{align}
    | S^{(l)}-S^{(l+1)}| &\le 2\sum_{i\in F_s}|\alpha_{i}^{(t+l)}||\beta_{i,i}^{(t+l)}|\\
    &\le  2 \sum_{i\in F_s}(|\alpha_{i}^{(t+l)}||\overline{\alpha}_{i}^{(t+l)}-\beta_{i,i}^{(t+l)}|+ |\alpha_{i}^{(t+l)}||\overline{\alpha}_{i}^{(t+l)}|)\\
    &\le  2 \sum_{i\in F_s}(|\alpha_{i}^{(t+l)}|\||\overline{\psi}_k^{(t+l)}\rangle-|\psi_{k+1,i}^{(t+l)}\rangle\|_2+ |\alpha_{i}^{(t+l)}||\overline{\alpha}_{i}^{(t+l)}|)\\
    &\le  4\sum_{i\in F_s}\sum_{j=0}^{l}|\alpha_{i}^{(t+l)}||\overline{\alpha}_{i}^{(t+j)}| && \text{(by Eq.~\eqn{bar-approx})}\\
    &\le 4\sum_{j=0}^l \sqrt{(\sum_{i\in F_s}|\alpha_{i}^{(t+l)}|^2 )(\sum_{i\in F_s}|\overline{\alpha}_{i}^{(t+j)}|^2)}\\
   & \le4(l+1).
\end{align}
Hence \begin{align}
    |S^{(0)}-S^{(m)}|\le 4\sum_{l=0}^{m-1}(l+1)\le 2m(m+1).
\end{align}
Since \begin{align}
    \Bigl|\sum_{i\in F_s}(\langle \psi_{k}^{(t)}|\psi_{k+1,i}^{(t)}\rangle-\langle \psi_{k}^{(t+m)}|\psi_{k+1,i}^{(t+m)}\rangle)\Bigr|=|S^{(0)}-S^{(m)}|\le 2m(m+1)
\end{align}
and $|F_s|=n-k$, there exists $i_0\in F_s$ such that \begin{align}
    \bigl|\langle \psi_{k}^{(t)}|\psi_{k+1,i_0}^{(t)}\rangle-\langle \psi_{k}^{(t+m)}|\psi_{k+1,i_0}^{(t+m)}\rangle\bigr|\le \frac{2m(m+1)}{n-k}.
\end{align}

Now we can construct the string sequences which are hard to distinguish. For $t=0$, previous arguments guarantee that there exists $i_0$ such that \begin{align}
    \bigl|\langle \psi_{k}^{(0)}|\psi_{k+1,i_0}^{(0)}\rangle-\langle \psi_{k}^{(m)}|\psi_{k+1,i_0}^{(m)}\rangle\bigr|\le \frac{2m(m+1)}{n-k}.
\end{align} Since $\langle \psi_{k}^{(0)}|\psi_{k+1,i_0}^{(0)}\rangle=1$, we have \begin{align}
    \bigl|\langle \psi_{k}^{(m)}|\psi_{k+1,i_0^{(1)}}^{(m)}\rangle\bigr|\ge 1- \frac{2m(m+1)}{n-k}.
\end{align}
Then let $|\psi_{k+1,i_0}^{(m)}\rangle$ be $|\psi_{k+1}^{(m)}\rangle$, we can find $i_0^{(2)}$ such that \begin{align}
    \bigl|\langle \psi_{k}^{(2m)}|\psi_{k+1,i_0^{(2)}}^{(2m)}\rangle\bigr|\ge 1-\frac{4m(m+1)}{n-k}.
\end{align}
Let $T=\frac{n-k}{4(m+1)^2}$. Repeat this process, we can find $i_0^{(1)},\ldots,i_0^{(T)}$ such that \begin{align}
    \bigl|\langle \psi_{k}^{(mT)}|\psi_{k+1,i_0^{(T)}}^{(mT)}\rangle\bigr|\ge \frac{1}{2}.
\end{align}
    
Therefore, the algorithm needs $\Omega(mT)=\Omega(\frac{n-k}{m})$ queries to distinguish between $|x|=k$ and $|x|=k+1$ with constant success probability.

Similarly, we can prove that any algorithm needs $\Omega(\frac{k}{m})$ queries to distinguish between $|x|=k$ and $|x|=k-1$ with constant success probability. Hence the algorithm needs $\Omega(\max(\frac{k}{m},\frac{n-k}{m}))=\Omega(\frac{n}{m})$ queries to distinguish between $|x|=k$ and $|x|=k+1$ or $|x|=k-1$ with constant success probability. Note that the above analysis holds when $T\ge1$, so $m$ needs to be $O(\sqrt{n})$.
\end{proof}


\begin{thebibliography}{10}

\bibitem{bennett1989time}
Charles~H. Bennett.
\newblock Time/space trade-offs for reversible computation.
\newblock {\em SIAM Journal on Computing}, 18(4):766--776, 1989.

\bibitem{boyer1998tight}
Michel Boyer, Gilles Brassard, Peter H{\o}yer, and Alain Tapp.
\newblock Tight bounds on quantum searching.
\newblock {\em Fortschritte der Physik: Progress of Physics}, 46(4-5):493--505,
  1998.

\bibitem{brassard2002quantum}
Gilles Brassard, Peter Hoyer, Michele Mosca, and Alain Tapp.
\newblock Quantum amplitude amplification and estimation.
\newblock {\em Contemporary Mathematics}, 305:53--74, 2002.

\bibitem{burchard2019lower}
Paul Burchard.
\newblock Lower bounds for parallel quantum counting.
\newblock {\em arXiv preprint arXiv:1910.04555}, 2019.

\bibitem{chakrabarti2021threshold}
Shouvanik Chakrabarti, Rajiv Krishnakumar, Guglielmo Mazzola, Nikitas
  Stamatopoulos, Stefan Woerner, and William~J. Zeng.
\newblock A threshold for quantum advantage in derivative pricing.
\newblock {\em Quantum}, 5:463, 2021.

\bibitem{cornelissen2023sublinear}
Arjan Cornelissen and Yassine Hamoudi.
\newblock A sublinear-time quantum algorithm for approximating partition
  functions.
\newblock In {\em Proceedings of the 2023 Annual ACM-SIAM Symposium on Discrete
  Algorithms (SODA)}, pages 1245--1264. SIAM, 2023.

\bibitem{dagum2000optimal}
Paul Dagum, Richard Karp, Michael Luby, and Sheldon Ross.
\newblock An optimal algorithm for {M}onte {C}arlo estimation.
\newblock {\em SIAM Journal on Computing}, 29(5):1484--1496, 2000.

\bibitem{dean2020sample}
Sarah Dean, Horia Mania, Nikolai Matni, Benjamin Recht, and Stephen Tu.
\newblock On the sample complexity of the linear quadratic regulator.
\newblock {\em Foundations of Computational Mathematics}, 20(4):633--679, 2020.

\bibitem{grover1996fast}
Lov~K Grover.
\newblock A fast quantum mechanical algorithm for database search.
\newblock In {\em Proceedings of the twenty-eighth annual ACM symposium on
  Theory of computing}, pages 212--219, 1996.

\bibitem{grover2004quantum}
Lov~K. Grover and Jaikumar Radhakrishnan.
\newblock Quantum search for multiple items using parallel queries.
\newblock {\em arXiv preprint quant-ph/0407217}, 2004.

\bibitem{hamoudi2021quantum}
Yassine Hamoudi.
\newblock Quantum sub-{G}aussian mean estimator.
\newblock In {\em 29th Annual European Symposium on Algorithms}, 2021.

\bibitem{hamoudi2019quantum}
Yassine Hamoudi and Fr{\'e}d{\'e}ric Magniez.
\newblock Quantum chebyshev's inequality and applications.
\newblock In {\em 46th International Colloquium on Automata, Languages, and
  Programming (ICALP 2019)}, 2019.

\bibitem{hoyer2007negative}
Peter Hoyer, Troy Lee, and Robert Spalek.
\newblock Negative weights make adversaries stronger.
\newblock In {\em Proceedings of the thirty-ninth annual ACM symposium on
  Theory of computing}, pages 526--535, 2007.

\bibitem{jeffery2017optimal}
Stacey Jeffery, Frederic Magniez, and Ronald De~Wolf.
\newblock Optimal parallel quantum query algorithms.
\newblock {\em Algorithmica}, 79:509--529, 2017.

\bibitem{kothari2023mean}
Robin Kothari and Ryan O'Donnell.
\newblock Mean estimation when you have the source code; or, quantum {Monte
  Carlo} methods.
\newblock In {\em Proceedings of the 2023 Annual ACM-SIAM Symposium on Discrete
  Algorithms (SODA)}, pages 1186--1215. SIAM, 2023.

\bibitem{montanaro2015quantum}
Ashley Montanaro.
\newblock Quantum speedup of {Monte Carlo} methods.
\newblock {\em Proceedings of the Royal Society A: Mathematical, Physical and
  Engineering Sciences}, 471(2181):20150301, 2015.

\bibitem{nielsen2001quantum}
Michael~A. Nielsen and Isaac~L. Chuang.
\newblock Quantum computation and quantum information.
\newblock {\em Phys. Today}, 54(2):60, 2001.

\bibitem{reichardt2009span}
Ben~W Reichardt.
\newblock Span programs and quantum query complexity: The general adversary
  bound is nearly tight for every boolean function.
\newblock In {\em 2009 50th Annual IEEE Symposium on Foundations of Computer
  Science}, pages 544--551. IEEE, 2009.

\bibitem{simchowitz2020naive}
Max Simchowitz and Dylan Foster.
\newblock Naive exploration is optimal for online {LQR}.
\newblock In {\em International Conference on Machine Learning}, pages
  8937--8948. PMLR, 2020.

\bibitem{vershynin2018high}
Roman Vershynin.
\newblock {\em High-dimensional probability: An introduction with applications
  in data science}, volume~47.
\newblock Cambridge university press, 2018.

\bibitem{yoder2014fixed}
Theodore~J. Yoder, Guang~Hao Low, and Isaac~L. Chuang.
\newblock Fixed-point quantum search with an optimal number of queries.
\newblock {\em Physical review letters}, 113(21):210501, 2014.

\bibitem{zalka1999grover}
Christof Zalka.
\newblock Grover’s quantum searching algorithm is optimal.
\newblock {\em Physical Review A}, 60(4):2746, 1999.

\bibitem{zhang2024parallel}
Zhicheng Zhang, Qisheng Wang, and Mingsheng Ying.
\newblock Parallel quantum algorithm for {Hamiltonian} simulation.
\newblock {\em Quantum}, 8:1228, 2024.

\end{thebibliography}
\end{document}